\documentclass{article}
\usepackage[utf8]{inputenc}
\usepackage[margin=1.2in]{geometry}

\usepackage{graphicx}
\usepackage{xcolor}
\usepackage{amsmath}
\usepackage{amssymb}
\usepackage{amsthm}
\usepackage{algorithm}
\usepackage{algpseudocode}
\usepackage{hyperref}
\usepackage{cleveref}
\usepackage[shortlabels]{enumitem}
\hypersetup{
     colorlinks = true,
     linkcolor = blue,
     citecolor = blue,
     filecolor = blue,
     urlcolor = blue
     }

\algnotext{EndFor}
\algnotext{EndIf}
\algnotext{EndWhile}
\algnotext{EndProcedure}

\allowdisplaybreaks

\newif\ificalp
\icalpfalse

\newif\iffullversion
\fullversiontrue

\iffullversion

\newtheorem{lemma}{Lemma}
\newtheorem{theorem}[lemma]{Theorem}
\newtheorem{claim}[lemma]{Claim}

\newtheorem{definition}[lemma]{Definition}

\newtheorem{corollary}[lemma]{Corollary}
\fi

\DeclareMathOperator{\Cov}{Cov}

\newcommand{\by}{\mathbf{y}}
\newcommand{\bx}{\mathbf{x}}

\newcommand{\tT}{\tilde{T}}
\newcommand{\tZ}{\tilde{Z}}

\newcommand{\eps}{\epsilon}
\newcommand{\E}{\mathbb{E}}

\newcommand{\cP}{\mathcal{P}}
\newcommand{\cM}{\mathcal{M}}

\newcommand{\cS}{\mathcal{S}}
\newcommand{\cT}{\mathcal{T}}
\newcommand{\cI}{\mathcal{I}}
\newcommand{\cZ}{\mathcal{Z}}
\newcommand{\bk}{\mathbf{k}}
\newcommand{\dks}{\textsc{DkS}}
\newcommand{\SOPTc}{S^{\OPT}_{\text{close}}}

\newcommand{\cB}[2]{\mathcal{B}(#1,#2)}
\newcommand{\bcB}[2]{\overline{\mathcal{B}(#1,#2)}}
\newcommand{\ind}{\mathbf{1}}
\newcommand{\oind}{\overline{\ind}}

\newcommand{\cX}{\mathcal{X}}
\newcommand{\bw}{\mathbf{w}}
\newcommand{\bW}{\mathbf{W}}
\newcommand{\ow}{\overline{\bw}}

\DeclareMathOperator{\argmax}{argmax}
\DeclareMathOperator{\DCG}{DCG}
\DeclareMathOperator{\OPT}{OPT}
\DeclareMathOperator{\disp}{Disp}
\DeclareMathOperator{\dive}{Div}
\DeclareMathOperator{\den}{Den}
\DeclareMathOperator{\poly}{poly}

\newcommand{\N}{\mathbb{N}}
\newcommand{\R}{\mathbb{R}}

\title{Improved Approximation Algorithms and Lower Bounds for Search-Diversification Problems}
\ificalp
\titlerunning{Improved Approximation and Lower Bounds for Search-Diversification}

\author{Anon}{Anon}{}{}{}

\fi

\iffullversion
\author{
Amir Abboud\\
Weizmann Institute\\
\texttt{amir.abboud@weizmann.ac.il}
\and
Vincent Cohen-Addad \\
Google Research \\ \texttt{cohenaddad@google.com}
\and
Euiwoong Lee\thanks{Partially supported by  Google.}\\
University of Michigan\\
\texttt{euiwoong@umich.edu}
\and
Pasin Manurangsi\\
Google Research\\
\texttt{pasin@google.com}
}
\fi

\begin{document}

\maketitle

\begin{abstract}
We study several questions related to diversifying search results. We give improved approximation algorithms in each of the following problems, together with some lower bounds.
\begin{enumerate}
\item We give a polynomial-time approximation scheme (PTAS) for a diversified search ranking problem~\cite{BansalJKN10} whose objective is to minimizes the discounted cumulative gain. Our PTAS runs in time $n^{2^{O(\log(1/\eps)/\eps)}} \cdot m^{O(1)}$ where $n$ denotes the number of elements in the databases and $m$ denotes the number of constraints. Complementing this result, we show that no PTAS can run in time $f(\eps) \cdot (nm)^{2^{o(1/\eps)}}$ assuming Gap-ETH and therefore our running time is nearly tight. Both our upper and lower bounds answer open questions from~\cite{BansalJKN10}.
\item We next consider the Max-Sum Dispersion problem, whose objective is to select $k$ out of $n$ elements from a database that maximizes the dispersion, which is defined as the sum of the pairwise distances under a given metric. We give a quasipolynomial-time approximation scheme (QPTAS) for the problem which runs in time $n^{O_{\eps}(\log n)}$. This improves upon previously known polynomial-time algorithms with approximate ratios 0.5~\cite{HassinRT97,BorodinJLY17}. Furthermore, we observe that reductions from previous work rule out approximation schemes that run in $n^{\tilde{o}_\eps(\log n)}$ time assuming ETH.
\item Finally, we consider a generalization of Max-Sum Dispersion called Max-Sum Diversification. In addition to the sum of pairwise distance, the objective also includes another function $f$. For monotone submodular function $f$, we give a quasipolynomial-time algorithm with approximation ratio arbitrarily close to $(1 - 1/e)$. This improves upon the best polynomial-time algorithm which has approximation ratio $0.5$~\cite{BorodinJLY17}. Furthermore, the $(1 - 1/e)$ factor is also tight as achieving better-than-$(1 - 1/e)$ approximation is NP-hard~\cite{Feige98}.
\end{enumerate}
\end{abstract}

\section{Introduction}

A fundamental task in databases in general and in search engines in particular is the selection and ordering of the results to a given query. 
Suppose that we have already retrieved the set of appropriate answers $S_q$ to a query $q$ by a certain preliminary process. Which item from the (possibly huge) set $S_q$ should be presented \emph{first}? Which should be the first ten?

Besides the obvious approach of ranking the \emph{most relevant} answers first, perhaps the second most important consideration is that the output set should satisfy certain \emph{diversity} requirements.
If a user searches for ``Barcelona'' it would be desirable that the first ten results contain a mix of items containing, e.g. general details of the city, tourist information, and news about the associated soccer team, even though the most relevant items in certain absolute terms may only pertain to the latter.
There are various natural ways to formalize what makes a set of results diverse, and 
much research has gone into this \emph{Search Diversification} topic in the past two and a half decades in various context (see e.g. \cite{CarbonellG98,AgrawalGHI09,GollapudiS09,BhaskaraGMS16,BansalJKN10,kulesza2012determinantal,rodrygo2015search,BorodinJLY17,DrosouJPS17,BasteJMPR19,FominGPP021,Moumoulidou0M21,HKKLO21,DBLP:conf/kdd/AbbassiMT13,DBLP:conf/pods/IndykMMM14,DBLP:conf/spaa/EpastoMZ19,DBLP:conf/aaai/ZadehGMZ17}).
Recently, there have also been extensive research efforts into algorithmic fairness (see e.g. a survey~\cite{fairness-survey}). Some of these fairness notions (e.g.~\cite{Chierichetti0LV17,BackursIOSVW19}) are also closely related to diversity:
a set of results that is not diverse enough (e.g. returning only pictures of members of one group when a user searches for ``scientists'') could be problematic in terms of fairness.

A well-known work on search diversification \cite{CarbonellG98} suggests that a diverse set of results is one that satisfies the following: The $k^{th}$ result in the list should maximize the sum\footnote{To be more precise, it is a weighted average of the two terms.} of: (1) the relevance to the query, and (2) the total distance to the first $k-1$ results in the list.
The success of this natural notion of diversification may be attributed to the fact that it can be computed efficiently with a greedy algorithm.
However, it may be a bit too simplistic and the objectives that real-world search engines seem to optimize for are actually closer to other, more complicated (to compute) notions of diversity that have been proposed in follow-up works (e.g. \cite{BansalJKN10,GollapudiS09,BorodinJLY17}).

The goal of this paper is to investigate the time complexity of computing these latter, more intricate definitions of the search diversification task.
Since such problems are NP-Hard even for restricted settings, and since approximate solutions are typically acceptable in this context, our focus is on understanding their time vs. approximation trade-offs.
Our results reduce the gaps in the literature, completely resolving the complexity of some of the most natural notions.

\subsection{Diversified Search Ranking}

The first problem we study is a diversified search ranking problem formulated by Bansal et al.~\cite{BansalJKN10}. Here we are given a collections $\cS$ of subsets of $[n]$ and, for each $S \in \cS$, a positive integer $k_S$. Our goal is to find a permutation $\pi: [n] \to [n]$ that minimizes the \emph{discounted cumulative gain (DCG)} defined as
\begin{align} \label{eq:dcg-def}
\DCG_{\cS, \bk}(\pi) := \sum_{S \in \cS} \frac{1}{\log(t_{\pi}(S) + 1)},
\end{align}
where $t_{\pi}(S)$ is defined as the earliest time the set $S$ is covered $k_S$ times, i.e. $\min \{i \in [n] |S \cap \pi([i])| \geq k_S\}$.

This formulation relates to diversification by viewing the output $\pi$ as the ranking of the documents to be shown, and each topic corresponds to a set $S$ of documents related to that topic. With this interpretation, the DCG favors rankings that display ``diverse topics as early in the ranking as possible''. Bansal et al.~\cite{BansalJKN10} gave a polynomial-time approximation scheme (PTAS) for the problem in the special case that $k_S = 1$ for all $S \in \cS$ with running time $n^{2^{O(\log(1/\eps)/\eps)}} m^{O(1)}$. On the other hand, for the case of general $k_S$'s, they give a quasipolynomial-time approximation scheme with running time $n^{(\log \log n)^{O(1/\eps)}} m^{O(1)}$ and left as an open question whether a PTAS exists. We resolve this open question by giving a PTAS for the more general problem; the running time we obtain for this more general problem is similar to the running time 
obtained by Bansal et al.'s PTAS for the special case $k_S = 1$. We then show that this 
is indeed the best possible (under some complexity assumption).

\begin{theorem} \label{thm:main-ptas}
There is a randomized PTAS for maximizing DCG that runs in time $n^{2^{O(\log(1/\eps)/\eps)}} \cdot m^{O(1)}$.
\end{theorem}

The above running time is doubly exponential in $1/\eps$, and Bansal et al.~\cite{BansalJKN10} asked whether this dependency is necessary even for the special case $k_S = 1$. We also answer this question by showing that the doubly exponential is necessary, assuming the Gap Exponential Time Hypothesis (Gap-ETH)\footnote{Gap-ETH~\cite{Din16,ManurangsiR17} asserts that there is no $2^{o(n)}$-time algorithm to distinguish between a satisfiable $n$-variable 3SAT formula and one which is not even $(1 - \eps)$-satisfiable for some $\eps > 0$}:
\begin{theorem} \label{thm:dcg-lb}
Assuming Gap-ETH, for any function $g$, there is no PTAS for maximizing DCG that runs in time $g(\eps) \cdot (nm)^{2^{o(1/\eps)}}$. Moreover, this holds even when restricted to instances with $k_S = 1$ for all $S \in \cS$.
\end{theorem}

\subsection{Max-Sum Dispersion}

The second problem we consider is the so-called \emph{Max-Sum Dispersion} problem where we are given a metric space $(U, d)$ where $|U| = n$ and an integer $p \geq 2$. The goal is to select $S \subseteq U$ of size $p$ that maximizes
\begin{align*}
\disp(S) := \sum_{\{u, v\} \subseteq S} d(u, v).
\end{align*}

Roughly speaking, if the metric determines how different the items are, then our goal is to pick items that are ``as diverse as possible'' according to the $\disp$ objective.

The Max-Sum Dispersion problem is a classic problem that has been studied since the 80s~\cite{MoonC84,Kuby87,RaviRT94,HassinRT97,BorodinJLY17}.
Previous works have given 0.5-approximation algorithm for the problem in polynomial time~\cite{HassinRT97,BorodinJLY17}. We observe that the known NP-hardness reduction, together with newer hardness of approximation results for the Densest $k$-Subgraph problem with perfect completeness, yields strong lower bounds for the problem. (\Cref{app:hardness-from-dks}.) For example, if we assume the Strongish Planted Clique Hypothesis~\cite{ManurangsiRS21}, then no $(0.5 + \eps)$-approximation algorithm is possible in $n^{o(\log n)}$ time. In other words, to achieve an improvement over the known approximation ratio, the algorithm must run in $n^{\Omega(\log n)}$ time. Complementing this, we provide a quasipolynomial-time approximation scheme that runs in time $n^{O_\eps(\log n)}$:
\begin{theorem} \label{thm:qptas-dispersion}
There is a QPTAS for Max-Sum Dispersion that runs in time $n^{O(\log n / \eps^4)}$.
\end{theorem}

\subsection{Max-Sum Diversification}

Finally, we consider a generalization of Max-Sum Dispersion where, in addition to the metric space $(U, d)$, we are now also given a monotone set function $f$ (which we can access via a value oracle) and the goal is to select a set $S \subseteq U$ of size $p$ that maximizes $$\dive(S) := \disp(S) + f(S).$$ This problem is referred to as \emph{Max-Sum Diversification}.

The Max-Sum Diversification problem is more expressive than Max-Sum Dispersion. For example, the value $f(S)$ in the objective may be used to encode how relevant the selected set $S$ is to the given query, in addition to the diversity objective expressed by $\disp(S)$. 

Borodin et al.~\cite{BorodinJLY17} gave a 0.5-approximation algorithm for the problem when $f$ is a monotone submodular function. Since Max-Sum Diversification is a generalization of Max-Sum Dispersion, our aforementioned lower bounds also imply that improving on this 0.5 factor requires at least $n^{\Omega(\log n)}$ time. Furthermore, submodular Max-Sum Diversification is also a generalization of maximizing monotone submodular function subject to a cardinality constraint. For this problem, an $(1 - 1/e)$-approximation algorithm is known and it is also known that achieving better than this ratio is NP-hard~\cite{Feige98}. Therefore, it is impossible to achieve a better-than-$(1 - 1/e)$ approximation even in (randomized) quasi-polynomial time, assuming NP $\nsubseteq RTIME(n^{O(\log n)})$. Here we manage to provide such a tight quasi-polynomial time approximation algorithm:

\begin{theorem} \label{thm:max-sum-div-apx}
For any $\eps > 0$, there exists a randomized $n^{O(\log n / \eps^4)}$-time $(1 - 1/e - \eps)$-approximation algorithm for submodular Max-Sum Diversification.
\end{theorem}

We remark that an interesting special case of submodular Max-Sum Diversification is when $f$ is linear, i.e. $f(S) = \sum_{u \in S} f(u)$. In this case, Gollapudi and Sharma~\cite{GollapudiS09} provided an approximation-preserving reduction from the problem to the Max-Sum Dispersion. Therefore, our QPTAS for the latter (\Cref{thm:qptas-dispersion}) also yields a QPTAS for this special case of Max-Sum Dispersion.

\section{Preliminaries}
For a natural number $n$, we use $[n]$ to denote $\{1, \dots, n\}$. We say that a randomized algorithm for a maximization problem is an $\alpha$-approximation if the expected objective of the output solution is at least $\alpha$ times the optimum; note that we can easily get a high-probability bound with approximation guarantee arbitrarily close to $\alpha$ by repeating the algorithm multiple times and pick the best solution.

\subsection{Concentration Inequalities}

For our randomized approximation algorithms, we will need some standard concentration inequalities. First, we will use the following version of Chernoff bound which gives a tail bound on the sum of i.i.d. random variables. (See e.g.~\cite{MitzenmacherU-book} for a proof.)

\begin{lemma}[Chernoff bound] \label{lem:chernoff}
Let $X_1, \dots, X_r \in [0, 1]$ be independent random variables, $S := X_1 + \cdots + X_r$ and $\mu := \E[S]$. Then, for any $\delta \in [0, 1]$, we have
\begin{align*}
\Pr[|S - \mu| > \delta \mu] \leq 2 \exp\left(-\frac{\delta^2 \mu}{3}\right).
\end{align*}
Furthermore, for any $\delta \geq 0$, we have
\begin{align*}
\Pr[S > (1 + \delta)\mu] \leq \exp\left(- \frac{\delta^2 \mu}{2 + \delta}\right).
\end{align*}
\end{lemma}

It will also be convenient to have a concentration of sums of random variables that are drawn without replacement from a given set. For this, we will use (a without-replacement version of) the Hoeffding's inequality, stated below. (See e.g.~\cite{bardenet2015concentration}.)

\begin{lemma}[Hoeffding's inequality] \label{lem:hoeffding}
Let $X_1, \dots, X_r$ be random variables drawn without replacement from a multiset $\cX \subseteq [0, 1]$, $A := \frac{1}{r}\left(X_1 + \cdots + X_r\right)$ and $\mu := \E[A]$. Then, for any $\delta \in [0, 1]$, we have
\begin{align*}
\Pr[|A - \mu| > \delta] \leq 2 \exp\left(-2\delta^2 r\right).
\end{align*}
\end{lemma}

\subsection{Densest $k$-Subgraph}

For both our Max-Sum Dispersion and Max-Sum Diversification problems, we will use as a subroutine algorithms for (variants of) the \emph{Densest $k$-Subgraph (\dks)} problem. In \dks, we are given a set $V$ of nodes, weights $w: \binom{V}{2} \to [0, 1]$ and an integer $k$, the goal is to find a subset $T \subseteq V$ with $|T| = k$ that maximizes $\den(T) := \frac{1}{|T|(|T|-1)/2} \sum_{\{u, v\} \subseteq T} w(\{u, v\})$. An {\em additive QPTAS} is an algorithm running in quasipolynomial time for any fixed $\eps > 0$ such that its output $T$ satisfies $\den(T) \geq \OPT - \eps$; Barman~\cite{Barman18} gave such an algorithm for \dks. 

We will in fact use a slightly generalized version of the problem where a subset $I \subseteq V$ of vertices is given as an input and these vertices must be picked in the solution $T$ (i.e. $I \subseteq T$). To avoid cumbersomeness, we also refer to this generalized version as \dks. It is not hard to see\footnote{In fact, in \Cref{subsec:submodular-dks}, we also give a more general algorithm than the one stated in~\Cref{thm:qptas-dks} which can also handle an additional monotone submodular function.} that Barman's algorithm~\cite{Barman18} extends easily to this setting:
\begin{theorem} \label{thm:qptas-dks}
There is an additive QPTAS for \dks\ that runs in time $n^{O(\log n / \eps^2)}$.
\end{theorem}

\dks\ is a classic problem in approximation algorithms literature, and many approximation algorithms~\cite{FS97,SW98,FL01,FPK01,AHI02,GL09,BCCFV10,Barman18} and hardness results~\cite{Feige02,Kho06,RS10,alon2011inapproximability,BCVGZ12,BravermanKRW17,Manurangsi17,ChalermsookCKLM20} have been proved over the years. Most of these works focus on \emph{multiplicative} approximation; the best known polynomial-time algorithm in this setting has an approximation ratio of $n^{1/4 + \eps}$ for any constant $\eps > 0$~\cite{BCCFV10} and there are evidences that achieving subpolynomial ratio in polynomial time is unlikely~\cite{Manurangsi17,BCVGZ12,CMMV17}. As for \emph{additive} approximation, it is known that an approximation scheme that runs in time $n^{\tilde{o}(\log n)}$ would break the exponential time hypothesis (ETH)~\cite{BravermanKRW17}; therefore, the running time in~\Cref{thm:qptas-dks} (in terms of $n$) is tight up to $\poly\log \log n$ factor in the exponent. We provide additional discussions on related results in~\Cref{app:hardness-from-dks}.

\subsection{Submodular Maximization over a Matroid Constraint}

For our approximation algorithm for Max-Sum Diversification, we will also need an approximation algorithm for \emph{monotone submodular maximization under a matroid constraint}. In this problem, we are given a monotone submodular set function $f: 2^X \to \R_{\geq 0}$ over a ground set $X$ together with a matroid $\cM = (X, \cI)$. The function $f$ is given via a value oracle and $\cM$ can be accessed via a membership oracle (which answers questions of the form ``does $S$ belong to $\cI$?''). The goal is to find $S \in \cI$ that maximizes $f(S)$. C{\u{a}}linescu et al. gave a randomized algorithm with approximation ratio $(1 - 1/e)$ for the problem, which we will use in our algorithm.

\begin{theorem}[\cite{CalinescuCPV11}] \label{thm:submodular-matroid}
There exists a randomized polynomial-time $(1 - 1/e)$-approximation algorithm for maximizing a montone submodular function over a matroid constraint.
\end{theorem}

\section{Diversified Search Ranking}

In this section, we consider the diversified search ranking question as proposed in~\cite{BansalJKN10} and prove our upper and lower bounds (\Cref{thm:main-ptas,thm:dcg-lb}).

\subsection{Polynomial-time Approximation Scheme}

We will start by presenting our PTAS. At a high-level, our PTAS is similar to that of Bansal et al.'s: our algorithm use bruteforce to try every possible values of $\pi(1), \dots, \pi(\exp(\tilde{O}(1/\epsilon)))$. Once these are fixed, we solve the remaining problem using linear programming (LP). We use the same LP as Bansal et al., except with a slightly more refined rounding procedure, which allows us to achieve a better approximation guarantee.

The remainder of this section is organized as follows. In \Cref{subsec:lp-rounding}, we present our LP rounding algorithm and its guarantees. Then, we show how to use it to yield our PTAS in \Cref{subsec:ptas-ranking}.

\subsubsection{Improved LP Rounding}
\label{subsec:lp-rounding}

For convenience in the analysis below, let us also define a more generic objective function where $\frac{1}{\log(t_\pi(S)) + 1}$ in~\Cref{eq:dcg-def} can be replaced by any non-increasing function $f: [n] \to (0, 1]$:
\begin{align*}
\DCG^f_{\cS, \bk}(\pi) := \sum_{S \in \cS} f(t_{\pi}(S)).
\end{align*}

The main result of this subsection is the following polynomial time LP rounding algorithm for the above general version of DCG:

\begin{lemma} \label{lem:rounding-final}
There exists an absolute constant $C$ such that for any $\alpha \in (0, 0.5)$ the following holds: there is a polynomial-time algorithm that computes a ranking with expected DCG at least $(1 - \alpha) \cdot \tau_{f, \alpha}$ times that of the optimum where
$$\tau_{f, \alpha} := \min_{t \in [n]} \frac{f\left(\frac{C \log(1/\alpha)}{\alpha} \cdot \frac{t}{f(t)}\right)}{f(t)}.$$
\end{lemma}

Informally speaking, the term $\tau_{f, \alpha}$ somewhat determines ``how fast $f$ increases''. In the next section, once we fix the first $u$ elements of the ranking, $f$ will become $f(t) := 1/\log(t + u)$ which is ``slowly growing'' when $u$ is sufficiently large. This allows us to ensure that the guarantee in \Cref{lem:rounding-final} yields an $(1 - O(\eps))$-approximation as desired.

\paragraph{LP Formulation.}
To prove \Cref{lem:rounding-final}, we use the same knapsack constraint-enhanced LP as in~\cite{BansalJKN10}, stated below. Note that the number of knapsack constraints can be super-polynomial. However, it is known that such an LP can be solved in polynomial time; see e.g.~\cite[Section 3.1]{BansalGK10} for more detail.
\begin{align*}
&\text{Maximize} & \sum_{S \in \cS} \sum_{t \in [n]} (y_{S, t} - y_{S, t - 1}) \cdot f(t) & & \\
&\text{subject to} &\sum_{e \in [n]} x_{e,t} = 1 & &\forall t \in [n] \\
& &\sum_{t \in [n]} x_{e,t} = 1 & &\forall e \in [n] \\
& &\sum_{e \in S \subseteq A} \sum_{t' < t} x_{e,t'} \geq (k_S - |A|) \cdot y_{S,t} & &\forall S \in \cS, A \subseteq S, t \in [n] \\
& &y_{S, t} \geq y_{S, t - 1} & &\forall S \in \cS, t \in \{2, \dots, n\} \\
& &x_{e,t}, y_{S,t} \in [0, 1] & &\forall e, t \in [n], S \in \cS.
\end{align*}

\paragraph{Rounding Algorithm.} Let $\gamma \in (0, 0.1)$ be a parameter to be chosen later. Our rounding algorithm works as follows:
\begin{enumerate}
\item $\pi \leftarrow \emptyset$
\item For $i = 1, \dots, \lceil\log n\rceil$ do:
\begin{enumerate}
\item Let $t_i = \min\{n, 2^i\}$.
\item Let $z_{e, i} = \sum_{t \leq t_i} x^*_{e, t}$ and $p_{e, i} = \min\{1, \frac{z_{e, i}}{\gamma \cdot f(t_i)}\}$ for all $e \in [n]$.
\item Let $A_i$ be the set such that $e \in [n]$ is independently included w.p. $p_{e, i}$.
\end{enumerate}
\end{enumerate}
Finally, our permutation $\pi$ is defined by adding elements from $A_1, \dots, A_{\lceil \log n\rceil}$ in that order, where the order within each $A_i$ can be arbitrary and we do not add an element if it already appears in the permutation.

Once again, we remark that our algorithm closely follows that of~\cite{BansalJKN10}, except that Bansal et al. simply chose their $p_{e, i}$ to be $\min\{1, O(\log^2 n) \cdot z_{e, i}\}$, whereas our $p_{e, i}$ is a more delicate $\min\{1, \frac{z_{e, i}}{\gamma \cdot f(t_i)}\}$. This allows our analysis below to produce a better approximation ratio.

\paragraph{Analysis.} We will now proceed to analyze our proposed randomized rounding procedure. Let $\eta \in (0, 0.1)$ be a parameter to be chosen later, and let $(\bx^*, \by^*)$ denote an optimal solution to the LP. For each $S$, let $t^*(S)$ be the largest positive integer $t^*$ such that
\begin{align} \label{eq:tstar-def}
y^*_{S,t^* - 1} \leq \eta \cdot f(t^*).
\end{align}

We start with the following lemma, which is a refinement of~\cite[Lemma 1]{BansalJKN10}.

\begin{lemma} \label{lem:opt-sep-bound}
$\OPT \leq (1 + \eta) \cdot \sum_{S \in \cS} f(t^*(S))$.
\end{lemma}

\begin{proof}
We have
\begin{align*}
\OPT &\leq \sum_{S \in \cS} \sum_{t \in [n]} (y^*_{S, t} - y^*_{S, t - 1}) \cdot f(t) \\
&= \sum_{S \in \cS} \left(\sum_{t=1}^{t^*(S)-1} (y^*_{S, t} - y^*_{S, t - 1}) \cdot f(t) + \sum_{t=t^*(S)}^{n} (y^*_{S, t} - y^*_{S, t - 1}) \cdot f(t)\right) \\
&\leq \sum_{S \in \cS} \left(\sum_{t=1}^{t^*(S)-1} (y^*_{S, t} - y^*_{S, t - 1})+ \sum_{t=t^*(S)}^{n} (y^*_{S, t} - y^*_{S, t - 1}) \cdot f(t^*(S))\right) \\
&\leq \sum_{S \in \cS} \left(y^*_{S, t^*(S) - 1} + f(t^*(S))\right) \\
&\overset{\eqref{eq:tstar-def}}{\leq}  \sum_{S \in \cS} (1 + \eta) \cdot f(t^*(S)). \qedhere
\end{align*}
\end{proof}

Next, we show via standard concentration inequalities that $|A_i|$'s has small sizes with a large probability.

\begin{lemma}
With probability $1 - 2\exp\left(-\frac{1}{3\gamma}\right)$, we have $|A_i| \leq \frac{2 t_i}{\gamma f(t^*)}$ for all $i \in [\lceil \log n\rceil]$.
\end{lemma}

\begin{proof}
Notice that $\sum_{e \in [n]} p_{e, i} \leq \frac{\sum_{e \in [n]} z_{e, i}}{\gamma f(t_i)} = \frac{t_i}{\gamma f(t_i)}$. As a result, by Chernoff bound (\Cref{lem:chernoff}), we have
\begin{align*}
\Pr\left[|A_i| > \frac{2 t_i}{\gamma f(t^*)}\right] \leq \exp\left(-\frac{t_i}{3\gamma f(t^*)}\right) \leq \exp\left(-\frac{t_i}{3\gamma}\right).
\end{align*}
By union bound, we thus have $|A_i| \leq \frac{2 t_i}{\gamma f(t^*)}$ for all $i \in [\lceil \log n\rceil]$ with probability at least
\begin{align*}
1 - \sum_{i \in [\lceil \log n\rceil]} \exp\left(-\frac{t_i}{3\gamma}\right) \leq 1 - 2\exp\left(-\frac{1}{3\gamma}\right).
\end{align*}
\end{proof}

Let $i^*(S)$ denote the smallest $i$ such that $t_i \geq t^*(S)$. We now bound the probability that $S$ is covered ($k_S$ times) by the end of the $i^*(S)$-th iteration of the algorithm. Our bound is stated below. We note that our bound here is not with high probability, unlike that of the analysis of~\cite{BansalJKN10} which yields a bound of $1 - o(1/n)$. We observe here that such a strong bound is not necessary for the analysis because we are working with a maximization problem and therefore such a high probability bound is not necessary to get a bound on the expectation of the DCG.

\begin{lemma}
Assume that $\eta \geq 2\gamma$.
For each $S \in \cS$, we have $t_{\pi}(S) \leq |A_1| + \cdots + |A_{i^*(S)}|$ with probability $1 - \exp\left(\frac{\eta}{8\gamma}\right)$.
\end{lemma}

\begin{proof}
It suffices to show that at least $k_S$ elements of $S$ are selected in $A_{i^*(S)}$. Let $S_g$ denote the set of elements $e \in S$ for which $p_{e, i^*(S)} = 1$. If $|S_g| \geq k_S$, then we are done. Otherwise, from knapsack constraint, we have
\begin{align*}
\sum_{e \in S \setminus S_g} z_{e, i^*(S)} \geq (k_S - |S_g|) y^*_{S, t_{i^*(S)}} \geq  (k_S - |S_g|) y^*_{S, t^*(S)} &\geq \eta \cdot f(t^*(S)) \cdot (k_S - |S_g|) \\
&\geq \eta \cdot f(t_{i^*(S)}) \cdot (k_S - |S_g|),
\end{align*}
where the third inequality follows from our choice of $t^*(S)$. This implies that
\begin{align*}
\sum_{e \in S \setminus S_g} p_{e, i^*(S)} \geq \eta / \gamma \cdot (k_S - |S_g|).
\end{align*}
Recall that $\eta / \gamma \geq 2$. This means that the probability that at least $k_S$ elements of $S$ are selected in $A_{i^*(S)}$ is at least
\begin{align*}
&1 - \Pr[|(S \setminus S_g) \cap A_{i^*(S)}| \leq 0.5\eta / \gamma \cdot (k_S - |S_g|)] \\
&\leq 1 - \exp\left(-\frac{1}{8} \cdot \eta / \gamma \cdot (k_S - |S_g|)\right) \\
&\leq 1 - \exp\left(-\frac{\eta}{8\gamma}\right),
\end{align*}
where the first inequality follows from the Chernoff bound.
\end{proof}

Applying the union bound to the two previous lemmas, we immediately arrive at the following:

\begin{lemma} \label{lem:term-by-term}
Assume that $\eta \geq 2\gamma$.
For all $S \in \cS$, we have
$$\E_\pi[f(t_\pi(S))] \geq \left(1 - 2\exp\left(-\frac{1}{3\gamma}\right) - \exp\left(\frac{\eta}{8\gamma}\right)\right) \cdot f\left(\frac{8 t^*(S)}{\gamma f(t^*(S))}\right)$$
\end{lemma}

Finally, combining~\Cref{lem:opt-sep-bound,lem:term-by-term} and selecting $\eta = 2\alpha, \gamma = O(\eta / \log(1/\eta))$ yields \Cref{lem:rounding-final}.

\subsubsection{From LP Rounding to PTAS}
\label{subsec:ptas-ranking}

As stated earlier, we may now use bruteforce to try all possible values of the first few elements in the ranking and then use our LP rounding to arrive at the PTAS:

\begin{proof}[Proof of \Cref{thm:main-ptas}]
For any $\eps < 0.1$, we use bruteforce for the first $u = (4C/\eps)^{100/\eps}$ elements and then use~\Cref{lem:rounding-final} on the remaining instance but with $f(t) := \frac{1}{\log(t + u)}$. The expected approximation ratio we have is at least
\begin{align*}
&(1 - 0.5\eps) \cdot \tau_{f, 0.5\eps} \\
&\geq (1 - 0.5\eps) \cdot \min_{t \in [n]} f\left(\frac{4C \log(1/\eps)}{\eps} \cdot \frac{t}{f(t)}\right) / f(t) \\
&= (1 - 0.5\eps) \cdot \min_{t \in [n]} \frac{\log(t + u)}{\log\left(\frac{4C \log(1/\eps)}{\eps} \cdot \frac{t}{f(t)} + u\right)} \\
&\geq (1 - 0.5\eps) \cdot \min_{t \in [n]} \frac{\log(t + u)}{\log\left(\frac{4C \log(1/\eps)}{\eps} \cdot (t + u)\log(t + u)\right)} \\
&= (1 - 0.5\eps) \cdot \min_{t \in [n]} \frac{1}{1 + \frac{\log\left(\frac{4C \log(1/\eps)}{\eps}\right)}{\log(t + u)} + \frac{\log\log(t + u)}{\log(t + u)}} \\
&= (1 - 0.5\eps) \cdot \frac{1}{1 + \frac{\log\left(\frac{4C \log(1/\eps)}{\eps}\right)}{\log(u)} + \frac{\log\log(u)}{\log(u)}} \\
&\geq (1 - 0.5\eps) \cdot \frac{1}{1 + 0.1\eps + 0.1\eps} \\
&\geq 1 - \eps,
\end{align*}
as desired.
\end{proof}

\subsection{Running Time Lower Bound}

To prove our running time lower bound, we will reduce from the \emph{Maximum $k$-Coverage} problem. Recall that in Maximum $k$-Coverage, we are given a set $\cT \subseteq [M]$ and an integer $k$; the goal is to find $T^*_1, \cdots T^*_k \in \cT$ that maximizes $|T^*_1 \cup \cdots \cup T^*_k|$. We write $\Cov(\cT, k)$ to denote this optimum. Furthermore, we say that a Maximum $k$-Coverage is \emph{regular} if $|T| = M/k$ for all $T \in \cT$. Finally, we use $N$ to denote $|\cT| \cdot M$ which upper bound the ``size'' of the problem.

Manurangsi~\cite{Manurangsi20} showed the following lower bound for this problem:
\begin{theorem}[\cite{Manurangsi20}] \label{thm:max-coverage-lb}
Assuming the Gap Exponential Time Hypothesis (Gap-ETH), for any constant $\delta > 0$, there is no $N^{o(k)}$-time algorithm that can, given a regular instance $(\cT, k)$  distinguish between the following two cases:
\begin{itemize}
\item (YES) $\Cov(\cT, k) \geq M$.
\item (NO) $\Cov(\cT, k) \leq (1 - 1/e + \delta) M$.
\end{itemize}
\end{theorem}

\begin{proof}[Proof of \Cref{thm:dcg-lb}]
Fix $\delta = 0.1$.
We reduce from the Maximum $k$-Coverage problem. Suppose that $(\cT, k)$ is a regular Maximum $k$-Coverage instance; we assume w.l.o.g. that $k$ is divisible by 10. 

We construct the instance $(\cS, \{k_S\}_{S \in \cS})$ of the DCG maximization as follows:
\begin{itemize}
\item Let $n = |\cT|$ where we associate each $j \in [n]$ with $T_j \in \cT$.
\item Let $\cS = \{S_1, \dots, S_M\}$ where $S_i = \{j \in [n] \mid i \in T_j\}$.
\item Let $k_S = 1$ for all $S \in \cS$.
\end{itemize}

In the YES case, let $T_{j_1}, \dots, T_{j_k}$ be such that $|T_{j_1} \cup \cdots \cup T_{j_k}| = M$. Let $\pi^*: [n] \to [n]$ be any permutation such that $\pi^*(\ell) = j_\ell$ for all $\ell \in [k]$. From regularity of $(\cT, k)$, there are exactly $q := M/k$ sets $S \in \cS$ such that $t_{\pi^*}(S) = i$. Therefore, we have
\begin{align*}
\DCG_{\cS, \bk}(\pi^*)
&= \sum_{i \in [k]} \frac{M}{k} \cdot \frac{1}{\log(i + 1)}.
\end{align*}
Let $\OPT^*$ denote the RHS quantity. Notice that \begin{align}
\OPT^* \leq \frac{M}{\log(k+1)}.
\end{align}

In the NO case, consider any permutation $\pi: [n] \to [n]$. Let $t_i$ denote the $i$-th smallest value in the multiset $\{t_{\pi}(S)\}_{S \in \cS}$. Regularity of $(\cT, k)$ implies that 
\begin{align} \label{eq:t-gap}
t_i \geq t_{i - q} + 1
\end{align} for all $i > q$. 
This in turn implies that
\begin{align} \label{eq:t-from-size}
t_i \geq \left\lceil i/q \right\rceil.
\end{align}
Furthermore, $\Cov(\cT, k) \leq (1 - 1/e - \delta)M \leq 0.8M$ implies that
\begin{align*}
t_{0.8M} > k.
\end{align*}
Furthermore, applying~\eqref{eq:t-gap} to the above, we have
\begin{align} \label{eq:t-from-unconver}
t_{0.9M} \geq t_{0.8M} + \left\lfloor\frac{0.1M}{q}\right\rfloor = k + 0.1k = 1.1k.
\end{align}

With the above notion, we may write $\DCG_{\cS, \bk}(\pi) - \OPT^*$ as
\begin{align*}
\DCG_{\cS, \bk}(\pi) - \OPT^* &= \sum_{i=1}^M \frac{1}{\log(t_i + 1)} - \sum_{i=1}^M \frac{1}{\log(\lceil i/q \rceil + 1)} \\
&\overset{\eqref{eq:t-from-size}}{\geq} \sum_{i=0.9M}^M \left(\frac{1}{\log(t_i + 1)} - \frac{1}{\log(\lceil i/q \rceil + 1)}\right) \\
&\overset{\eqref{eq:t-from-unconver}}{\geq}  \sum_{i=0.9M}^M \left(\frac{1}{\log(1.1k + 1)} - \frac{1}{\log(\lceil i/q \rceil + 1)}\right) \\
&\geq \sum_{i=0.9M}^M \left(\frac{1}{\log(1.1k + 1)} - \frac{1}{\log(k + 1)}\right) \\
&= 0.1M \cdot \left(\frac{1}{\log(1.1k + 1)} - \frac{1}{\log(k + 1)}\right) \\
&= \Theta\left(\frac{M}{\log^2 k}\right).
\end{align*}
Finally, observe also that
\begin{align*}
\OPT^* = \frac{M}{k} \cdot \sum_{i\in [k]} \frac{1}{\log(i+1)} = \frac{M}{k} \Theta\left(\frac{k}{\log k}\right) = \Theta\left(\frac{M}{\log k}\right).
\end{align*}
Combining the above two inequalities, we have
\begin{align*}
\DCG_{\cS, \bk}(\pi) \geq \left(1 + \Theta\left(\frac{1}{\log k}\right)\right) \cdot \OPT^*.
\end{align*}

Now, suppose that there is a PTAS for maximizing DCG that runs in time $f(\eps) \cdot (nm)^{2^{o(1/\eps)}}$. If we run the algorithm with $\eps = \gamma / \log k$ where $\gamma > 0$ is sufficiently small constant, then we can distinguish between the YES case and the NO case in time $f(1/\log k) \cdot (nm)^{2^{o(\log k)}} \leq f(1/\log k) \cdot (nm)^{o(k)} = g(k) \cdot N^{o(k)}$ which, from \Cref{thm:max-coverage-lb}, violates Gap-ETH.
\end{proof}

\section{Max-Sum Dispersion}
\label{sec:disp}

In this section, we provide a QPTAS for Max-Sum Dispersion (\Cref{thm:qptas-dispersion}).


As alluded to earlier, our algorithm will reduce to the Densest $k$-Subgraph (\dks) problem, for which an \emph{additive} QPTAS is known~\cite{Barman18}.
Notice here that \dks\ is a generalization of the Max-Dispersion problem because we may simply set $V = U, k = p$ and $w(\{u, v\}) = d(u, v) / D$ where $D := \max_{u, v} d(u, v)$ denote the diameter of the metric space. Note however that we cannot apply \Cref{thm:qptas-dks} yet because the QPTAS in that theorem offers an \emph{additive} guarantee. E.g. if the optimum is $o(1)$, then the QPTAS will not yield anything at all unless we set $\eps = o(1)$, which then gives a running time $n^{\omega(\log n)}$. This example can happen when e.g. there is a single pair $u, v$ that are very far away and then all the other pairs are close to $u$.

Our main technical contribution is to give a simple structural lemma that allows us to avoid such a scenario. Essentially speaking, it allows us to pick a vertex and selects all vertices that are ``too far away'' from it. Once this is done, the remaining instance can be reduced to \dks\ without encountering the ``small optimum'' issue described in the previous paragraph.






\subsection{A Structural Lemma}

Henceforth, we write $\disp(S, T)$ to denote $\sum_{u \in S, v \in T} d(u, v)$ and $\disp(u, T)$ as a shorthand for $\disp(\{u\}, T)$. Furthermore, we use $\cB{u}{D}$ to denote $\{z \in U \mid d(z, u) \leq D\}$ and let $\bcB{u}{D} := U \setminus \cB{u}{D}$.

We now formalize our structural lemma. It gives a lower bound on the objective based on a vertex in the optimal solution and another vertex \emph{not} in the optimal solution. Later on, by guessing these two vertices, we can reduce to \dks\ while avoiding the ``small optimum'' issue.

\begin{lemma} \label{lem:max-dispersion-structural}
Let $S^{\OPT}$ be any optimal solution of Max-Sum Dispersion and let $u^{\min}$ be the vertex in $S^{\OPT}$ that minimizes $\disp(u^{\min}, S^{\OPT})$. Furthermore, let $v$ be any vertex \emph{not} in $S^{\OPT}$ and let $\Delta = d(u^{\min}, v)$. Then, we have
\begin{align*}
\disp(S^{\OPT}) \geq \frac{p(p - 1)\Delta}{16}.
\end{align*} 
\end{lemma}

\begin{proof}[Proof of \Cref{lem:max-dispersion-structural}]
Let $\SOPTc := S^{\OPT} \cap \cB{u^{\min}}{0.5\Delta}$. Consider two cases, based on the size of $\SOPTc$:
\begin{itemize}
\item Case I: $|\SOPTc| \leq p / 2$. In this case, we have 
\begin{align*}
\disp(u^{\min}, S^{\OPT}) \geq \disp(u^{\min}, S^{\OPT} \setminus \SOPTc) \geq (p/2)(\Delta/2) = \Delta p / 4.
\end{align*}
Furthermore, by our definition of $u^{\min}$, we have
\begin{align*}
\disp(S^{\OPT}) = \frac{1}{2} \sum_{u \in S} \disp(u, S^{\OPT}) \geq \frac{p}{2} \disp(u^{\min}, S^{\OPT}).
\end{align*}
Combining the two inequalities, we have $\disp(S^{\OPT}) \geq p^2\Delta / 8$.
\item Case II: $|\SOPTc| > p / 2$. In this case, since $S^{\OPT}$ is an optimal solution, replacing any $z \in S^{\OPT}_{\text{close}}$ with $v$ must not increase the solution value, i.e. 
\begin{align*}
\disp(z, S^{\OPT}) &\geq \disp(v, S^{\OPT} \setminus \{z\}) \\
&\geq \disp(v, \SOPTc \setminus \{z\}) \\
&\geq ((p - 1)/2)(0.5\Delta),
\end{align*}
where the second inequality uses the fact that for any $z' \in \SOPTc$ we have $d(v, z') \geq d(u, v) - d(u, z') \geq \Delta - 0.5\Delta$. From this, we once again have
\begin{align*}
\disp(S^{\OPT}) = \frac{1}{2} \sum_{u \in S} \disp(u, S^{\OPT}) \geq \frac{1}{2} \sum_{z \in \SOPTc} \disp(z, S^{\OPT}) &\geq |\SOPTc| \cdot \frac{(p - 1)\Delta}{8} \\
&> \frac{p(p - 1)}{16\Delta},
\end{align*}
where the last inequality follows from our assumption of this case. \qedhere
\end{itemize}
\end{proof}

\subsection{QPTAS for Max-Sum Dispersion}

We now present our QPTAS, which simply guesses $u^{\min}$ and $v = \argmax_{z \notin S^{\OPT}} d(z, u)$ and then reduces the problem to \dks. By definition of $v$, if we let $\Delta = d(u, v)$, every point outside $\cB{u^{\min}}{\Delta}$ must be in $S^{OPT}$. 
The actual reduction to \dks\ is slightly more complicated than that described at the beginning of this section. Specifically, among points $\bcB{u^{\min}}{\Delta}$ that surely belong to $S^{\OPT}$, we ignore all points outside $\cB{u^{\min}}{ 20\Delta/\eps}$ (i.e., they do not appear in the \dks\ instance) and we let $\cB{u^{\min}}{20\Delta/\eps} \setminus \cB{u^{\min}}{\Delta}$ be the ``must pick'' part. Ignoring the former can be done because the contribution to the objective from those points can be approximated to within $(1 \pm O(\eps))$ regardless of the points picked in the ball $\cB{u^{\min}}{\Delta}$. This is not true for the latter, which means that we need to include them in our \dks\ instance.

\begin{proof}[Proof of \Cref{thm:qptas-dispersion}]
Our algorithm works as follows:
\begin{enumerate}
\item For every distinct $u, v \in U$ do:
\begin{enumerate}
\item Let $\Delta := d(u, v)$ and $\Delta^* = 20 \Delta / \eps$.
\item If $|\bcB{u}{\Delta}| \geq p$, then skip the following steps and continue to the next pair $u, v$.
\item Otherwise, create a \dks\ instance where $V := \cB{u}{\Delta^*}, I := V \setminus \cB{u}{\Delta}$, $k = p - |\bcB{u}{\Delta^*}|$ and $w$ is defined as $w(\{y, z\}) := 0.5 d(y, z) / \Delta^*$ for all $y, z \in V$.
\item Use the additive QPTAS from \Cref{thm:qptas-dks} to solve the above instance to within an additive error of $\eps' := 0.00005\eps^2$. Let $T$ be the solution found.
\item Finally, let $S^{u, v} := T \cup \bcB{u}{\Delta^*}$.
\end{enumerate}
\item Output the best solution among $S^{u, v}$ considered.
\end{enumerate}

It is obvious that the running time is dominated by the running time of the QPTAS which takes $n^{O(\log n / (\eps')^2)} = n^{O(\log n / \eps^4)}$ as desired.

Next, we show that the algorithm indeed yields a $(1 - \eps)$-approximation. To do this, let us consider $S^{\OPT}, u^{\min}$ as defined in \Cref{lem:max-dispersion-structural}, and let $u = u^{\min}, v := \argmax_{z \notin S^{\OPT}} d(u, z)$. Let $T$ be the solution found by the \dks\ algorithm for this $u, v$ and let $T' := T \setminus I$. We have
\begin{align}
&\disp(S^{u, v}) \nonumber \\
&= \disp(\bcB{u}{\Delta^*}) + \disp(\bcB{u}{\Delta^*}, T) + \disp(T) \nonumber \\
&= \disp(\bcB{u}{\Delta^*}) + \disp(\bcB{u}{\Delta^*}, I) + \disp(\bcB{u}{\Delta^*}, T') + \disp(T). \label{eq:current-s-decompose}
\end{align}
Similarly, letting $S := S^{\OPT} \cap \cB{u}{\Delta^*}$ and $S' := S^{\OPT} \setminus I$, we have
\begin{align}
&\disp(S^{\OPT}) \nonumber \\
&= \disp(\bcB{u}{\Delta^*}) + \disp(\bcB{u}{\Delta^*}, I) + \disp(\bcB{u}{\Delta^*}, S') + \disp(S). \label{eq:optimal-s-decompose}
\end{align}

Now, observe from the definition of the \dks\ instance (for this $u, v$) that for any $J$ such that $I \subseteq J \subseteq V$, we have
\begin{align*} 
\den(J) = \frac{1}{k(k - 1)/2} \cdot \frac{0.5}{\Delta^*} \disp(J).
\end{align*}
The additive approximation guarantee from \Cref{thm:qptas-dks} implies that $\den(T) \geq \den(S) - \eps'$. Using the above equality, we can rewrite this guarantee as
\begin{align} \label{eq:dks-qptas-guarantee}
\disp(S) - \disp(T) \leq \eps' \cdot \Delta^* \cdot k(k - 1).
\end{align}

Taking the difference between~\Cref{eq:optimal-s-decompose} and~\Cref{eq:current-s-decompose} and applying~\Cref{eq:dks-qptas-guarantee}, we have
\begin{align*}
\disp(S^{\OPT}) - \disp(S^{u, v})
&\leq \disp(\bcB{z}{\Delta^*}, S') - \disp(\bcB{z}{\Delta^*}, T') + \eps' \cdot \Delta^* \cdot k(k - 1). \\
(\text{Our choice of } \eps') &\leq \disp(\bcB{z}{\Delta^*}, S') - \disp(\bcB{z}{\Delta^*}, T') + 0.001\eps \Delta \cdot p(p - 1) \\
(\text{\Cref{lem:max-dispersion-structural}}) &\leq \disp(\bcB{z}{\Delta^*}, S') - \disp(\bcB{z}{\Delta^*}, T') + 0.1\eps \disp(S^{\OPT}).
\end{align*}

Now, since $|S'| = |T'| \leq p$ and $S', T' \subseteq \cB{z}{\Delta}$, we have
\begin{align*}
\disp(\bcB{z}{\Delta^*}, S') - \disp(\bcB{z}{\Delta^*}, T') 
&\leq |\bcB{z}{\Delta^*}| \cdot |S'| \cdot ((\Delta^* + \Delta) - (\Delta^* - \Delta)) \\
&\leq 2 |\bcB{z}{\Delta^*}| \cdot |S'| \cdot \Delta \\
(\text{Our choice of } \Delta^*) &\leq 0.1\eps \cdot |\bcB{z}{\Delta^*}| \cdot |S'| \cdot (\Delta^* - \Delta) \\
&\leq 0.1\eps \cdot \disp(\bcB{z}{\Delta^*}, S') \\
&\leq 0.1\eps \cdot \disp(S^{\OPT}).
\end{align*}

Combining the above two inequalities, we get $\disp(S^{u, v}) \geq (1 - 0.2\eps) \cdot \disp(S^{\OPT})$, as desired.
\end{proof}

\section{Max-Sum Diversification}
\label{sec:diversification}
In this section, we give our quasipolynomial-time approximation algorithm for the Max-Sum Diversification with approximation ratio arbitrarily close to $(1 - 1/e)$ (\Cref{thm:max-sum-div-apx}). In fact, we prove a slightly stronger version of the theorem where the approximation ratio for the dispersion part is arbitrarily close to 1 and that of the submodular part is arbitrarily close to $1 - 1/e$. This is stated more precisely below; note that this obviously implies \Cref{thm:max-sum-div-apx}.

\begin{theorem} \label{thm:diversification-main-detailed}
Let $S^{\OPT}$ be any optimal solution of Max-Sum Diversification. There exists a randomized $n^{O(\log n / \eps^4)}$-time algorithm that finds a $p$-size set $S$ such that
\begin{align*}
\E[\dive(S)] \geq (1 - \eps) \disp(S^{\OPT}) + (1 - 1/e - \eps) f(S^{\OPT}).
\end{align*}
\end{theorem}

At a high-level, our algorithm for Max-Sum Diversification is very similar to that of Max-Sum Dispersion presented in the previous section. Specifically, we use a structural lemma (akin to \Cref{lem:max-dispersion-structural}) to reduce our problem to a variant of \dks. This variant of \dks\ additionally has a submodular function attached to it. Using techniques from \dks\ approximation literature, we give an algorithm for this problem by in turn reducing it to the submodular maximization problem over a partition matroid, for which we can appeal to \Cref{thm:submodular-matroid}.

\subsection{Approximating Densest Subgraph and Submodular Function}
\label{subsec:submodular-dks}
We will start by giving an algorithm for the aforementioned extension of the \dks\ problem, which we call \emph{Submodular \dks}:

\begin{definition}[Submodular \dks]
Given $(V, I, w, k)$ (similar to \dks) together with a monotone submodular set function $h$ on the ground set $V$ (accessible via a value oracle), the goal is to find a size-$k$ subset $T$ where $I \subseteq T \subseteq V$ that maximizes $h(T) + \den(T)$.
\end{definition}

We give a quasipolynomial-time algorithm with an approximation guarantee similar to QPTAS for the original \dks\ (i.e. \Cref{thm:qptas-dks}) while also achiving arbritrarily close to $(1 - 1/e)$ approximation ratio for the submodular part of the objective:

\begin{theorem} \label{thm:submodular-dks}
For any set $T^{\OPT}$ of size $k$ such that $I \subseteq T^{\OPT} \subseteq V$, there is an $n^{O(\log n / \gamma^2)}$-time algorithm that output a size-$k$ $T$ such that $I \subseteq T \subseteq V$ and 
\begin{align} \label{eq:apx-guarantee-submodular-dks}
\E[h(T) + \den(T)] \geq \left(1 - 1/e - \gamma\right) h(T^{\OPT}) + \den(T^{\OPT}) - \gamma.
\end{align}
\end{theorem}

In order to facilitate the subsequent discussion and proof, it is useful to define additional notations. (Throughout, we view vectors as column vectors.)
\begin{itemize}
\item Let $\bW \in \R^{V \times V}$ denote the vector where $\bW_{u, v} = w(\{u, v\})$ for $u \ne v$ and $\bW_{u, u} = 0$.
\item For every $U \subseteq V$, let $\ind(U) \in \R^V$ denote the indicator vector of $U$, i.e. 
\begin{align*}
\ind(U)_v =
\begin{cases}
1 &\text{ if } v \in U, \\
0 &\text{ otherwise.}
\end{cases}
\end{align*} 
\item For every $U \subseteq V$, let $\bw(U) = \bW \cdot \ind(U) \in \R^V$.
\item Finally, for every non-empty $U \subseteq V$, let $\ow(U) := \frac{1}{|U|} \cdot w(U)$ and $\oind(U) := \frac{1}{|U|} \cdot \ind(U)$.
\end{itemize}

To understand our reduction, we must first describe the main ideas behind the QPTAS of~\cite{Barman18}. (Some of these ideas also present in previous works, e.g.~\cite{AlonLSV13}.) Let us assume for simplicity of presentation that $I = \emptyset$. Observe that \dks\ is, up to an appropriate scaling, equivalent to find a size-$k$ subset $T$ that maximizes $\oind(T)^T \cdot \bW \cdot \oind(T) = \oind(T)^T \ow(T)$. The main observation is that, if we randomly pick a subset $U \subseteq T^{\OPT}$ of size $\Theta_\gamma(\log n)$, then with high probability $\|\ow(U) - \ow(T^{\OPT})\|_\infty \leq O(\gamma)$ and $|\oind(T^{\OPT})^T \ow(T^{\OPT}) - \oind(U)^T \ow(U)| < O(\gamma)$. Roughly speaking,~\cite{Barman18} exploits this by ``guessing'' such a set $U$ and then solves for $T$ such that $\|\ow(U) - \ow(T)\|_\infty \leq O(\gamma)$ and $|\oind(T)^T \ow(U) - \oind(U)^T \ow(U)| < O(\gamma)$; note that (the fractional version of) this is a linear program and can be solved efficiently. \cite{Barman18} then shows that a fractional solution to such a linear program can be rounded to an actual size-$k$ set without any loss in the objective function.

We further push this idea by noting that, if we randomly partition $V$ into $V_1, \dots, V_s$ part where $s = O_\gamma(k/\log n)$, then the intersections $U^{\OPT}_i := V_i \cap T^{\OPT}$ satisfy the two conditions from the previous paragraphs (for $T = U^{\OPT}_i$). Therefore, we may enumerate all sets $U_i \subseteq V_i$ of roughly expected size to construct a collection $\cP_i$ of subsets that satisfies these two conditions. Our goal now become picking $U_1 \in \cP_1, \dots, U_s \in \cP_s$ that maximizes $h(U_1 \cup \cdots \cup U_s)$. This is simply monotone submodular maximization subject to a partition matroid constraint and therefore we may appeal to \Cref{thm:submodular-matroid}. We remark here that the two conditions that all subsets in $\cP_i$ satisfy already ensure that the \dks\ objective is close to optimum.
\ificalp
The full proof of \Cref{thm:submodular-dks} is deferred to the appendix.
\fi

\iffullversion
The approach outlined above is formalized in the following proof. Note that the exact algorithm below is slightly more complicated than above since we also have to deal with the fact that $I$ may be non-empty.

\begin{proof}[Proof of \Cref{thm:submodular-dks}]
Let $k' := k - |I|, V' := V \setminus I, \gamma' = 0.01 \gamma$, $s := \lfloor 0.001 \gamma'^2 k' / \log n \rfloor, t := k'/s$. We may assume w.l.o.g. that $s \geq 1$; otherwise, we can easily solve the problem exactly in claimed running time via brute-force search. 

Our algorithm works as follows:
\begin{itemize}
\item Randomly partition $V'$ into $(V'_1, \dots, V'_s)$ where each vertex is independently place in each partition with probability $1/s$. 
\item For every non-empty subset $Q \subseteq V'$ of size at most $(1 + \gamma')t$, do:
\begin{itemize}
\item For $i = 1, \dots, s$:
\begin{itemize}
\item Let $\cP_i \leftarrow \emptyset$
\item For each non-empty subset $U_i \subseteq V'_i$ of size between $(1 - \gamma')t$ and $(1 + \gamma')t$:
\begin{itemize}
\item  If the following two conditions hold, then add $U_i$ to $\cP_i$:
\begin{align} \label{eq:cond-deg-concen}
\left\|\ow(U_i) - \ow(Q)\right\|_\infty \leq 2\gamma',
\end{align}
\begin{align} \label{eq:cond-opt-concen}
|\oind(U_i)^T \ow(Q) - \oind(Q)^T \ow(Q)| \leq 4\gamma'.
\end{align}
\end{itemize}
\end{itemize}
\item Create a partition matroid $\cM$ on the ground set $\cP_1 \cup \cdots \cup \cP_s$ such that $\cS \subseteq \cP_1 \cup \cdots \cup \cP_s$ is an independent set iff $|\cS \cap \cP_i| \leq 1$ for all $i \in [s]$.
\item Let $f$ denote the set function on the ground set $\cP_1 \cup \cdots \cup \cP_s$ defined as $h(\cS) := h\left(I \cup \bigcup_{S \in \cS} S\right)$
\item Run the algorithm from \Cref{thm:submodular-matroid} to on $(f, \cM)$ to get a set $\cZ_Q \subseteq \cP_1 \cup \cdots \cup \cP_s$.
\item Let $\tZ_Q = \bigcup_{S \in \cZ} S$.
\item If $|\tZ_Q| \geq k'$, let $\tT_Q$ be a random subset of $\tZ_Q$ of size $k'$. Otherwise, let $Z_Q$ be an arbitrary superset of $\tZ_Q$ of size $k'$.
\end{itemize}
\item Output the best set $I \cup Z_Q$ found among all $Q$'s.
\end{itemize}

It is obvious to see that the algorithm runs in $n^{O(t)} = n^{O(\log n / \gamma^2)}$ time. The rest of the proof is devoted to proving~\eqref{eq:apx-guarantee-submodular-dks}.

Let $T'^{\OPT} := T^{\OPT} \setminus I$, and let $U_i^{\OPT} := V'_i \cap T^{\OPT}$. We will start by proving the following claim, which (as we will argue below) ensures that w.h.p. $U_i^{\OPT}$ is included in $\cP_i$.

\begin{claim} \label{claim:random-partition-conditions}
With probability $1 - O(1/n)$ (over the random partition $V'_1, \dots, V'_s$), the following holds for all $i \in [s]$:
\begin{align} \label{eq:intersection-size}
|U_i^{\OPT}| \in [(1 - \gamma')t, (1 + \gamma')t]
\end{align}
\begin{align} \label{eq:weight-apx}
\left\|\ow(U_i) - \ow(T'^{\OPT})\right\|_\infty \leq \gamma'
\end{align}
\begin{align} \label{eq:opt-large}
|\oind(U_i)^T \ow(T'^{\OPT}) - \oind(T'^{\OPT})^T \ow(T'^{\OPT})| \leq \gamma'
\end{align}
\end{claim}

\begin{proof}[Proof of \Cref{claim:random-partition-conditions}]
We will argue that all conditions holds for a fixed $i \in [s]$ with probability $1 - O(1/n^2)$. Union bound over all $i \in [s]$ then yields the claim. 

Let us fix $i \in [s]$. Since each vertex is included in $U_i$ with probability $1/s$, we may apply Chernoff bound (\Cref{lem:chernoff}) to conclude that
\begin{align} \label{eq:intersection-size-fixed-i}
\Pr[|U_i^{\OPT}| \notin [(1 - \gamma')t, (1 + \gamma')t]] \leq 2 \exp\left(-\frac{\gamma'^2 t}{3}\right) \leq 2/n^3.
\end{align}

Next, consider a fixed $v \in V$. We will now bound the probability that $|\ow(U_i)_v - \ow(T'^{\OPT})_v| < \gamma'$. To do so, let us condition on the size of $U_i^{\OPT}$ equal to $g \in \N$. After such a conditioning, we may view the set $U_i^{\OPT}$ as being generated by drawing $u_1, \dots, u_g$ randomly without replacement from $T'^{\OPT}$. Since $\ow(U_i^{\OPT})_v = \frac{1}{g}\left(\sum_{i \in [g]} w(\{u_i, v\})\right)$ and $\E[\ow(U_i^{\OPT})_v] = \ow(T'^{\OPT})_v$, we may apply \Cref{lem:hoeffding} to conclude that
\begin{align} \label{eq:deg-concen-fixed-i}
\Pr[|\ow(U_i^{\OPT})_v - \ow(T'^{\OPT})_v| > \gamma' \mid |U_i^{\OPT}| = g] \leq 2\exp(-\gamma'^2 g). 
\end{align}
Therefore, we have
\begin{align*}
&\Pr[|\ow(U_i^{\OPT})_v - \ow(T'^{\OPT})_v| > \gamma'] \\
&\leq \Pr[|U_i^{\OPT}| < (1 - \gamma')t] + \Pr[|\ow(U_i^{\OPT})_v - \ow(T'^{\OPT})_v| > \gamma' \mid |U_i^{\OPT}| \geq (1 - \gamma')t] \\
&\overset{\text{\eqref{eq:intersection-size-fixed-i}, \eqref{eq:deg-concen-fixed-i}}}{\leq} 2/n^3 + 2\exp(-\gamma'^2 (1 - \gamma')t) \\
&\leq 4/n^3.
\end{align*}
Taking the union bound over all $v \in V$, we have
\begin{align*}
\Pr[\|\ow(U_i^{\OPT})_v - \ow(T'^{\OPT})_v\|_{\infty} > \gamma'] \leq 4/n^2.
\end{align*}

Analogous arguments also imply that
\begin{align*}
\Pr[|\oind(U_i)^T \ow(T'^{\OPT}) - \oind(T'^{\OPT})^T \ow(T'^{\OPT})| > \gamma'] \leq O(1/n^2).
\end{align*}

Applying the union bound, we conclude that all three conditions hold for a fixed $i$ with probability at least $1 - O(1/n^2)$. Finally, applying the union bound over all $i \in [s]$, we have that all three conditions hold for all $i \in [s]$ with probability at least $1 - O(1/n)$, which concludes our proof.
\end{proof}

Let $E$ denote the event that all conditions in \Cref{claim:random-partition-conditions} hold for all $i$. Conditioned on $E$, and letting $Q = U_1$. For all $i \in [s]$, we have
\begin{align*}
\left\|\ow(U_i) - \ow(Q)\right\|_\infty \leq \left\|\ow(U_i) - \ow(T'^{\OPT})\right\|_\infty + \left\|\ow(T'^{\OPT}) - \ow(Q)\right\|_\infty \overset{\text{\eqref{eq:weight-apx}}}{\leq} 2\gamma'.
\end{align*}
and
\begin{align*}
&|\oind(U_i)^T \ow(Q) - \oind(Q)^T \ow(Q)| \\
&\leq |\oind(U_i)^T \ow(T'^{\OPT}) - \oind(T'^{\OPT})^T \ow(T'^{\OPT})| + |\oind(T'^{\OPT})^T \ow(T'^{\OPT}) - \oind(Q)^T \ow(T'^{\OPT})| \\
&\qquad + |\oind(U_i)^T(\ow(Q) - \ow(T'^{\OPT}))| + |\oind(Q)^T(\ow(T'^{\OPT}) - \ow(Q))| \\
&\overset{\text{\eqref{eq:opt-large}}}{\leq} 2\gamma' + \|\oind(U_i)\|_1 \|\ow(Q) - \ow(T'^{\OPT})\|_\infty + \|\oind(Q)\|_1 \|\ow(Q) - \ow(T'^{\OPT})\|_\infty \\
&\overset{\text{\eqref{eq:weight-apx}}}{\leq} 4\gamma'.
\end{align*}
Therefore, $U_i^{\OPT}$ is included in $\cP_i$ for all $i \in [s]$. Thus, the guarantee of \Cref{thm:submodular-matroid} means that $\E[f(\cZ_{Q})] \geq (1 - 1/e)f(\{U_1^{\OPT}, \dots, U_s^{\OPT}\})$. This is equivalent to 
\begin{align} \label{eq:submod-not-final}
\E[h(I \cup \tZ_Q)] \geq h(T^{\OPT}).
\end{align}
Next, we will lower bound $\den(\tZ_Q)$. Note that we may assume\footnote{If $\cZ_Q \cap \cP_i = \emptyset$, we may add any element of $\cP_i$ to $\cZ_Q$.} that $\cZ_Q = \{U_1, \dots, U_s\}$ where $U_i \in \cP_i$. 

For sets $A, B \subseteq V$, we write $w(A)$ to denote $\sum_{\{u, v\} \subseteq A} w(\{u, v\})$ and $w(A, B)$ to denote $\sum_{u \in A, v \in B} w(\{u, v\})$.
We may write $w(I \cup \tZ_Q)$ as
\begin{align}
w(I \cup \tZ_Q)
&= w(I \cup U_1 \cup \dots \cup U_s) \nonumber \\
&= w(I) + \sum_{i \in [s]} w(I, U_i) + \frac{1}{2} \sum_{i, j \in [s]} w(U_i, U_j). \label{eq:weight-expand}
\end{align}
We will now lower bound each term in the sum. First, we have
\begin{align*}
w(I, U_i) &= \ind(I)^T W \ind(U_i) \\
&= \ind(I)^T \bw(U_i) \\
&= |I| \cdot |U_i| \cdot \oind(I)^T \ow(U_i) \\
&\overset{\eqref{eq:cond-deg-concen}}{\geq} |I| \cdot |U_i| \cdot \left(\oind(I)^T \ow(Q) - 2\gamma'\right) \\
&\overset{\eqref{eq:weight-apx}}{\geq} |I| \cdot |U_i| \cdot \left(\oind(I)^T \ow(T'^{\OPT}) - 3\gamma'\right).
\end{align*}
Recall from our construction that $|U_i| \geq (1 - \gamma')t = (1 - \gamma') |T'^{\OPT}|/s$. Plugging this into the above, we have
\begin{align} \label{eq:ui-i-lb}
w(I, U_i) \geq \frac{1 - \gamma'}{s} \cdot w(I, T'^{\OPT}) - 3\gamma' \cdot |I| \cdot |U_i|.
\end{align}
 
Secondly, we have
\begin{align*}
w(U_i, U_j) &= \ind(U_i)^T W \ind(U_j) \\
&= \ind(U_i)^T \bw(U_j) \\
&= |U_i| \cdot |U_j| \cdot \oind(U_i)^T \ow(U_j) \\
&\overset{\eqref{eq:cond-deg-concen}}{\geq} |U_i| \cdot |U_j| \cdot \left(\oind(U_i)^T \ow(Q) - 2\gamma'\right) \\
&\overset{\eqref{eq:cond-opt-concen}}{\geq} |U_i| \cdot |U_j| \cdot \left(\oind(Q)^T \ow(Q) - 6\gamma'\right) \\
&\overset{\eqref{eq:weight-apx}}{\geq} |U_i| \cdot |U_j| \cdot \left(\oind(Q)^T \ow(T'^{\OPT}) - 7\gamma'\right) \\
&\overset{\eqref{eq:opt-large}}{\geq} |U_i| \cdot |U_j| \cdot \left(\oind(T'^{\OPT})^T \ow(T'^{\OPT}) - 8\gamma'\right) \\
&= |U_i| \cdot |U_j| \cdot \left(\frac{2}{|T'^{\OPT}|^2} \cdot w(T'^{\OPT}) - 8\gamma'\right)
\end{align*}
Similar to above, we may now use the fact that $|U_i|, |U_j| \geq (1 - \gamma') |T'^{\OPT}|/s$ to derive
\begin{align} \label{eq:ui-uj-lb}
w(U_i, U_j) \geq \frac{2(1 - 2\gamma')}{s^2} w(T'^{\OPT}) - 8\gamma' \cdot |U_i| \cdot |U_j|
\end{align}

Plugging~\eqref{eq:ui-i-lb} and~\eqref{eq:ui-uj-lb} back into~\eqref{eq:weight-expand}, we arrive at
\begin{align*}
w(I \cup \tZ_Q) 
&\geq w(I) + (1 - \gamma') w(I, T'^{\OPT}) + (1 - 2\gamma') w(T'^{\OPT}) - 8 \gamma' |I \cup \tZ_Q|^2 \\
&\geq (1 - 2\gamma') w(S^{OPT}) - 8 \gamma' |I \cup \tZ_Q|^2.
\end{align*}
Therefore,
\begin{align*}
\den(I \cup \tZ_Q) = \frac{1}{\binom{|I \cup \tZ_Q|}{2}} \cdot w(I \cup \tZ_Q) \geq \frac{1 - 2\gamma'}{\binom{|I \cup \tZ_Q|}{2}} \cdot w(T^{\OPT}) - 16\gamma'.
\end{align*}
Since $|U_i| \leq (1 + \gamma')t$, we have $|\tZ_Q| \leq (1 + \gamma')k'$. This implies that $|I \cup \tZ_Q| \leq (1 + \gamma')k$ Therefore,
\begin{align} 
\den(I \cup \tZ_Q) 
&\geq \frac{k(k - 1)}{(1+\gamma')k((1+\gamma')k + 1)} \cdot \frac{1 - 2\gamma'}{\binom{k}{2}} \cdot w(T^{\OPT}) - 16\gamma' \nonumber \\
&\geq (1 - 5\gamma') \den(T^{\OPT}) - 16\gamma'. \label{eq:den-almost-final}
\end{align}

From \eqref{eq:submod-not-final} and \eqref{eq:den-almost-final}, we can conclude that 
\begin{align}
\E[\den(\tZ_Q \cup I) + h(\tZ_Q \cup I) | E] \geq (1 - 5\gamma') \den(T^{\OPT}) - 16\gamma' + (1 - 1/e)h(T^{\OPT}).
\end{align}
Recall from \Cref{claim:random-partition-conditions} that $E$ happens with probability at least $1 - O(1/n)$. Thus, we have
\begin{align} \label{eq:objective-bound-almost-final}
\E[\den(\tZ_Q \cup I) + h(\tZ_Q \cup I)] \geq  (1 - 5\gamma' - O(1/n)) \den(T^{\OPT}) - 16\gamma' + (1 - 1/e - O(1/n))h(T^{\OPT}).
\end{align}

Finally, $|\tZ_Q| \leq (1 + \gamma')k'$ also implies that
\begin{align*}
&\E[\den(Z_Q \cup I) + h(Z_Q \cup I)] \\
&\geq \frac{k'(k' - 1)}{(1 + \gamma')k'((1 + \gamma')k' - 1)} \E[\den(\tZ_Q \cup I)] + \frac{k'}{(1 + \gamma')k'} \E[h(\tZ_Q \cup I)] \\
&\geq (1 - 3\gamma') \E\left[\left(\den(\tZ_Q \cup I) + h(\tZ_Q \cup I) \right)\right] \\
&\overset{\eqref{eq:objective-bound-almost-final}}{\geq} (1 - 8\gamma' - O(1/n)) \den(T^{\OPT}) - 16\gamma' + (1 - 1/e - 3\gamma' - O(1/n))h(T^{\OPT}) \\
&\geq \den(T^{\OPT}) + (1 - 1/e - 3\gamma' - O(1/n))h(T^{\OPT}) - 24\gamma' - O(1/n),
\end{align*}
which is at least $\den(T^{\OPT}) + (1 - 1/e - \gamma)h(T^{\OPT}) - \gamma$ for sufficiently large $n \geq \Omega(1/\gamma)$. Note that when $n$ is $O(1/\gamma)$, we may simply run the bruteforce $2^{O(n)} = 2^{O(1/\gamma)}$ to solve the problem.
\end{proof}
\fi

\subsection{From Submodular \dks\ to Max-Sum Diversification}

\ificalp
Having provided an approximation algorithm for Submodular \dks, we can use it to approximate Max-Sum Diversification via a similar approach to the reduction from Max-Sum Dispersion to \dks\ in the previous section. In particular, we can prove a structural lemma for Max-Sum Diversification that is analogous to \Cref{lem:max-dispersion-structural} for Max-Sum Dispersion. We can then use the reduction nearly identical to the one in the proof of \Cref{thm:qptas-dispersion} to arrive at \Cref{thm:diversification-main-detailed}. The full details are deferred to the appendix.
\fi

\iffullversion
Having provided an approximation algorithm for Submodular \dks, we now turn our attention back to how to use it to approximate Max-Sum Diversification.

\subsubsection{A Structural Lemma}

We start by proving a structural lemma for Max-Sum Diversification that is analogous to \Cref{lem:max-dispersion-structural} for Max-Sum Dispersion.

\begin{lemma} \label{lem:max-diversification-structural}
Let $S^{\OPT}$ be any optimal solution of Max-Sum Diversification and let $u^{\min}$ be the vertex in $S^{\OPT}$ that minimizes $\disp(u^{\min}, S^{\OPT})$. Furthermore, let $v$ be any vertex \emph{not} in $S^{\OPT}$ and let $\Delta = d(u^{\min}, v)$. Then, we have
\begin{align*}
\dive(S^{\OPT}) \geq \frac{p(p - 1)\Delta}{16}.
\end{align*} 
\end{lemma}

\begin{proof}[Proof of \Cref{lem:max-diversification-structural}]
Let $\SOPTc := S^{\OPT} \cap \cB{u^{\min}}{0.5\Delta}$. Consider two cases, based on the size of $\SOPTc$:
\begin{itemize}
\item Case I: $|\SOPTc| \leq p / 2$. This is similar to the first case in the proof of \Cref{lem:max-dispersion-structural}: we have 
\begin{align*}
\disp(u^{\min}, S^{\OPT}) \geq \disp(u^{\min}, S^{\OPT} \setminus \SOPTc) \geq (p/2)(\Delta/2) = \Delta p / 4.
\end{align*}
Furthermore, by our definition of $u^{\min}$, we have
\begin{align*}
\disp(S^{\OPT}) = \frac{1}{2} \sum_{u \in S} \disp(u, S^{\OPT}) \geq \frac{p}{2} \disp(u^{\min}, S^{\OPT}).
\end{align*}
Combining the two inequalities, we have $\dive(S^{\OPT}) \geq \disp(S^{\OPT}) \geq p^2\Delta / 8$.
\item Case II: $|\SOPTc| > p / 2$. In this case, since $S^{\OPT}$ is an optimal solution, replacing any $z \in S^{\OPT}_{\text{close}}$ with $v$ must not increase the solution value, i.e. 
\begin{align*}
\left[f(S^{\OPT}) - f(S^{\OPT} \setminus \{z\})\right] + \disp(z, S^{\OPT}) &\geq \disp(v, S^{\OPT} \setminus \{z\}) \\
&\geq \disp(v, \SOPTc \setminus \{z\}) \\
&\geq ((p - 1)/2)(0.5\Delta),
\end{align*}
where the second inequality uses the fact that for any $z' \in \SOPTc$ we have $d(v, z') \geq d(u, v) - d(u, z') \geq \Delta - 0.5\Delta$. From this, we have
\begin{align*}
\dive(S^{\OPT}) = f(S^{\OPT}) + \disp(S^{\OPT}) 
&= f(S^{\OPT}) + \frac{1}{2} \sum_{u \in S} \disp(u, S^{\OPT}) \\
&\geq \sum_{z \in S^{\OPT}} \left[f(S^{\OPT}) - f(S^{\OPT} \setminus \{z\})\right] + \frac{1}{2} \sum_{u \in S} \disp(u, S^{\OPT}) \\
&\geq \frac{1}{2} \sum_{z \in \SOPTc} \left(\left[f(S^{\OPT}) - f(S^{\OPT} \setminus \{z\})\right] + \disp(z, S^{\OPT})\right) \\
&\geq |\SOPTc| \cdot \frac{(p - 1)\Delta}{8} \\
&> \frac{p(p - 1)}{16\Delta},
\end{align*}
where the last inequality follows from our assumption of this case. \qedhere
\end{itemize}
\end{proof}

\subsubsection{Putting Things Together: Proof of \Cref{thm:diversification-main-detailed}}

Finally, we use the structural lemma to reduce Max-Sum Diversification to Submodular \dks. Again, this reduction is analogous to that of Max-Sum Dispersion to \dks\ presented in the previous section.

\begin{proof}[Proof of \Cref{thm:diversification-main-detailed}]
Our algorithm works as follows:
\begin{enumerate}
\item For every distinct $u, v \in U$ do:
\begin{enumerate}
\item Let $\Delta := d(u, v)$ and $\Delta^* = 20 \Delta / \eps$.
\item If $|\bcB{u}{\Delta}| \geq p$, then skip the following steps and continue to the next pair $u, v$.
\item Otherwise, create a submodular \dks\ instance where $V := \cB{u}{\Delta^*}, I := V \setminus \cB{u}{\Delta}$, $k = p - |\bcB{u}{\Delta^*}|$, define $h$ by $h(C) := \frac{1}{k(k - 1) \Delta^*} \cdot f(\bcB{u}{\Delta} \cup C)$, and define $w$ by $w(\{y, z\}) := (0.5 / \Delta^*) \cdot d(y, z)$ for all $y, z \in V$.
\item Use the algorithm from \Cref{thm:submodular-dks} to solve the above instance with $\gamma := 0.00005\eps^2$. Let $T$ be the solution found.
\item Finally, let $S^{u, v} := T \cup \bcB{u}{\Delta^*}$.
\end{enumerate}
\item Output the best solution among $S^{u, v}$ considered.
\end{enumerate}

It is obvious that the running time is dominated by the running time of the algorithm from \Cref{thm:submodular-dks} which takes $n^{O(\log n / \gamma^2)} = n^{O(\log n / \eps^4)}$ as desired.

Next, we prove the algorithm's approximation guarantee. To do this, let us consider $S^{\OPT}, u^{\min}$ as defined in \Cref{lem:max-dispersion-structural}, and let $u = u^{\min}, v := \argmax_{z \notin S^{\OPT}} d(u, z), \Delta = d(u, v)$. Recall that by the definition of $v$, we have $S^{\OPT} \supseteq \bcB{u}{\Delta} = \bcB{u}{\Delta^*} \cup I$. 
Let $T$ be the solution found by the submodular \dks\ algorithm for this $u, v$ and let $T' := T \setminus I$. We have
\begin{align}
&\disp(S^{u, v}) \nonumber \\
&= \disp(\bcB{u}{\Delta^*}) + \disp(\bcB{u}{\Delta^*}, T) + \disp(T) \nonumber \\
&= \disp(\bcB{u}{\Delta^*}) + \disp(\bcB{u}{\Delta^*}, I) + \disp(\bcB{u}{\Delta^*}, T') + \disp(T). \label{eq:current-s-decompose-2}
\end{align}
Similarly, letting $S := S^{\OPT} \cap \cB{u}{\Delta^*}$ and $S' := S^{\OPT} \setminus I$, we have
\begin{align}
&\disp(S^{\OPT}) \nonumber \\
&= \disp(\bcB{u}{\Delta^*}) + \disp(\bcB{u}{\Delta^*}, I) + \disp(\bcB{u}{\Delta^*}, S') + \disp(S). \label{eq:optimal-s-decompose-2}
\end{align}

Now, observe from the definition of the submodular \dks\ instance (for this $u, v$) that for any $J$ such that $I \subseteq J \subseteq V$, we have
\begin{align*} 
\den(J) = \frac{1}{k(k - 1)\Delta^*} \cdot \disp(J)
\end{align*}
and
\begin{align*}
h(J) = \frac{1}{k(k - 1)\Delta^*} \cdot f(\bcB{u}{\Delta} \cup J).
\end{align*}
The approximation guarantee from \Cref{thm:submodular-dks} ensures that $\E[h(S^{u, v}) + \den(T)] \geq (1 - 1/e - \gamma)h(S^{\OPT}) + \den(S) - \gamma$. Using the above two equalities, we can rewrite this guarantee as
\begin{align} \label{eq:dks-submodular-guarantee}
\disp(S) + (1 - 1/e - \gamma)f(S^{\OPT}) - \E[\disp(T) + f(S^{u, v})] \leq \gamma \cdot \Delta^* \cdot k(k - 1).
\end{align}

Taking the difference between~\eqref{eq:optimal-s-decompose-2} and~\eqref{eq:current-s-decompose-2} and applying~\eqref{eq:dks-submodular-guarantee}, we have
\begin{align*}
&\disp(S^{\OPT}) + (1 - 1/e - \gamma)f(S^{\OPT}) - \E[\disp(S^{u, v}) + f(S^{u, v})] \\
&\leq \disp(\bcB{z}{\Delta^*}, S') - \disp(\bcB{z}{\Delta^*}, T') + \gamma \cdot 2\Delta^* \cdot k(k - 1). \\
(\text{Our choice of } \gamma) &\leq \disp(\bcB{z}{\Delta^*}, S') - \disp(\bcB{z}{\Delta^*}, T') + 0.001\eps \Delta \cdot p(p - 1) \\
(\text{\Cref{lem:max-diversification-structural}}) &\leq \disp(\bcB{z}{\Delta^*}, S') - \disp(\bcB{z}{\Delta^*}, T') + 0.1\eps \dive(S^{\OPT}).
\end{align*}

Now, since $|S'| = |T'| \leq p$ and $S', T' \subseteq \cB{z}{\Delta}$, we have
\begin{align*}
\disp(\bcB{z}{\Delta^*}, S') - \disp(\bcB{z}{\Delta^*}, T') 
&\leq |\bcB{z}{\Delta^*}| \cdot |S'| \cdot ((\Delta^* + \Delta) - (\Delta^* - \Delta)) \\
&\leq 2 |\bcB{z}{\Delta^*}| \cdot |S'| \cdot \Delta \\
(\text{Our choice of } \Delta^*) &\leq 0.1\eps \cdot |\bcB{z}{\Delta^*}| \cdot |S'| \cdot (\Delta^* - \Delta) \\
&\leq 0.1\eps \cdot \disp(\bcB{z}{\Delta^*}, S') \\
&\leq 0.1\eps \cdot \disp(S^{\OPT}).
\end{align*}

Combining the above two inequalities, we get
\begin{align*}
\E[\dive(S^{u, v})] 
&\geq \disp(S^{\OPT}) + (1 - 1/e - \gamma)f(S^{\OPT}) - 0.2\eps \dive(S^{\OPT}) \\
&\geq (1 - \eps) \disp(S^{\OPT}) + (1 - 1/e - \eps) f(S^{\OPT}),
\end{align*}
completing our proof.
\end{proof}
\fi

\section{Conclusion}

In this work, we consider three problems related to diversification: DCG in diversified search ranking, Max-Sum Dispersion and Max-Sum Diversification. For DCG, we give a PTAS and prove a nearly matching running time lower bound. For Max-Sum Dispersion, we give a QPTAS and similarly provide evidence for nearly matching running time lower bounds. Finally, we give a quasi-polynomial time algorithm for Max-Sum Diversification that achieves an approximation ratio arbitrarily close to $(1 - 1/e)$, which is also tight given the $(1 - 1/e + o(1))$ factor NP-hardness of approximating Maximum $k$-Coverage~\cite{Feige98}. Our algorithms for DCG and Max-Sum Diversification are randomized and it remains an interesting open question whether there are deterministic algorithms with similar running times and approximation ratios.

\iffullversion
\section*{Acknowledgment}

We are grateful to Karthik C.S. for insightful discussions, and to Badih Ghazi for encouraging us to work on the problems.
\fi

\ificalp
\bibliographystyle{plainurl}
\fi
\iffullversion
\bibliographystyle{alpha}
\fi
\bibliography{ref}

\appendix

\section{Inapproximability of Max-Sum Dispersion}
\label{app:hardness-from-dks}

As we have already seen in our algorithm (\Cref{sec:disp}), Max-Sum Dispersion problem is closely related to the Densest $k$-Subgraph problem. In fact, the known NP-hardness reductions (e.g.~\cite{MoonC84}) can also be viewed as a reverse reduction--form Densest $k$-Subgraph to Max-Sum Dispersion.

To make this formal, let us define the \emph{Densest $k$-Subgraph (\dks) with perfect completeness} be the same as \dks\ except that there is a promise that the optimum is exactly one, i.e. there exists a set $T \subseteq V$ of size $k$ such that $w(u, v) = 1$ for all distinct $u, v \in T$. With this definition, the known reductions may be formulated as follows:

\begin{lemma} \label{lem:hardness-disp-from-dks}
If there exists an $(0.5 + \delta)$-approximation algorithm for Max-Sum Dispersion that runs in time $f(n, p)$, then there is a $2\delta$-approximation algorithm for \dks\ with perfect completeness that runs in time $f(n, k)$.
\end{lemma}

\begin{proof}
Given an instance $(V, w, k)$ of \dks\ with perfect completeness. We create the instance $(U, d)$ of Max-Sum Dispersion as follows. Let $V = U, p = k$ and then let $d(u, v) = 1 + w(\{u, v\})$ for all distinct $u, v \in V$. It is simple to see that $d$ is a valid metric. Furthermore, for any $T \subseteq V$ of size $k$, we have $\disp(T) = \frac{k(k - 1)}{2} \cdot (1 + w(T))$. Since we know that the optimum of \dks\ instance is one, the optimum of Max-Sum Dispersion is $k(k - 1)$. Thus, an $(0.5 + \delta)$-approximation algorithm for Max-Sum Dispersion will find $T \subseteq V$ of size $k$ such that $\disp(T) \geq k(k - 1)(0.5 + \delta)$ which implies that $w(T) \geq 2\delta$. In other words, the algorithm yields a $2\delta$-approximation for \dks\ with perfect completeness.
\end{proof}

A number of recent works have proved hardness of approximation for \dks\ with perfect completeness, with varying inapproximability factors, running time lower bound and approximation ratio. Each of them results in a different hardness result for Max-Sum Dispersion, which we list below.

\paragraph*{Hardness Based on Planted Clique Problem.}
We start with the hardness results based on the \emph{Planted Clique problem}~\cite{Karp76,Jerrum92}, which is to distinguish between a random Erdos-Renyi $G(n, 1/2)$ graph and one in which a clique of size say $n^{0.4}$ is added. The \emph{Planted Clique Hypothesis} states that this problem cannot be solved in polynomial time. Alon et al.~\cite{alon2011inapproximability} showed that, under this hypothesis, \dks\ with perfect completeness is hard to approximate to any constant factor. Plugging this into \Cref{lem:hardness-disp-from-dks}, we immediately get:
\begin{corollary}
Assuming the Planted Clique Hypothesis, there is no polynomial-time $(0.5 + \delta)$-approximation algorithm for Max-Sum Dispersion for all constant $\delta > 0$.
\end{corollary}

Recently, a stronger hypothesis called the \emph{Strongish Planted Clique Hypothesis} has been proposed~\cite{ManurangsiRS21}. It states that the Planted Clique problem cannot be solved even in $n^{o(\log n)}$ time. Under this hypothesis, the above running time lower bound immediately improves to $n^{o(\log n)}$.
\begin{corollary}
Assuming the Strongish Planted Clique Hypothesis, there is no $n^{o(\log n)}$-time $(0.5 + \delta)$-approximation algorithm for Max-Sum Dispersion for all constant $\delta > 0$.
\end{corollary}

\paragraph*{Hardness Based on Exponential Time Hypotheses.}
Next, we state the hardness results based on the Exponential Time Hypothesis (ETH) and the Gap Exponential Time Hypothesis (Gap-ETH). ETH~\cite{IPZ01} postulates that there is no $2^{o(n)}$-time algorithm to decide whether a given $n$-variable 3SAT formula is satisfiable. Gap-ETH~\cite{Din16,ManurangsiR17} is a strengthening of ETH; it asserts that there is no $2^{o(n)}$-time algorithm to distinguish between a satisfiable $n$-variable 3SAT formula and one which is not even $(1 - \eps)$-satisfiable for some $\eps > 0$. 

Braverman et al.~\cite{BravermanKRW17} showed that, assuming Gap-ETH\footnote{Note that there reduction also works with ETH but it only gives a slightly weaker running time lower bound of $n^{\tilde{o}\left(\frac{\log n}{(\log \log n)^{O(1)}}\right)}$.}, there is no $n^{o(\log n)}$-time $(1 - \gamma)$-approximation algorithm for \dks\ with perfect completeness for some constant $\gamma > 0$. Plugging this into \Cref{lem:hardness-disp-from-dks}, it gives the following hardness for Max-Sum Dispersion.
\begin{corollary}
Assuming Gap-ETH, there is no $n^{\tilde{o}(\log n)}$-time $(1 - \gamma)$-approximation algorithm for Max-Sum Dispersion for constant some $\gamma > 0$.
\end{corollary}

Manurangsi~\cite{Manurangsi17} proved an ETH-based inapproximability result for \dks\ with perfect completeness that rules out any polynomial-time nearly-polynomial factor approximation. Plugging this to \Cref{lem:hardness-disp-from-dks}, we get:
\begin{corollary}
Assuming ETH, there is no polynomial-time $\left(0.5 + \frac{1}{n^{1/(\log \log n)^{O(1)}}}\right)$-approximation algorithm for Max-Sum Dispersion.
\end{corollary}

Finally, in addition to approximation algorithms, parameterized algorithms form another popular set of techniques used to handle hard problems (see e.g.~\cite{CyganFKLMPPS15} for more background). It is therefore natural to ask whether there exists a fixed-parameter tractable (FPT) algorithm for Max-Sum Dispersion that beats a factor of 0.5. Unfortunately, Chalermsook et al.~\cite{ChalermsookCKLM20} proved that, under Gap-ETH, there is no FPT algorithm for \dks\ with perfect completeness that achieves approximation ratio $k^{o(1)}$. Plugging this into \Cref{lem:hardness-disp-from-dks} also rules out FPT algorithm for Max-Sum Dispersion with better-than-0.5 approximation ratio:
\begin{corollary}
Assuming Gap-ETH, for any function $g$, there is no $g(k) \cdot n^{O(1)}$-time $\left(0.5 + \frac{1}{k^{o(1)}}\right)$-approximation algorithm for Max-Sum Dispersion.

\end{corollary}

\ificalp
\section{Missing Proofs from Section~\Cref{sec:diversification}}

\subsection{Proof of \Cref{thm:submodular-dks}}

The proof below follows the outline in \Cref{subsec:submodular-dks}. Note that the exact algorithm below is slightly more complicated than above since we also have to deal with the fact that $I$ may be non-empty.

\begin{proof}[Proof of \Cref{thm:submodular-dks}]
Let $k' := k - |I|, V' := V \setminus I, \gamma' = 0.01 \gamma$, $s := \lfloor 0.001 \gamma'^2 k' / \log n \rfloor, t := k'/s$. We may assume w.l.o.g. that $s \geq 1$; otherwise, we can easily solve the problem exactly in claimed running time via brute-force search. 

Our algorithm works as follows:
\begin{itemize}
\item Randomly partition $V'$ into $(V'_1, \dots, V'_s)$ where each vertex is independently place in each partition with probability $1/s$. 
\item For every non-empty subset $Q \subseteq V'$ of size at most $(1 + \gamma')t$, do:
\begin{itemize}
\item For $i = 1, \dots, s$:
\begin{itemize}
\item Let $\cP_i \leftarrow \emptyset$
\item For each non-empty subset $U_i \subseteq V'_i$ of size between $(1 - \gamma')t$ and $(1 + \gamma')t$:
\begin{itemize}
\item  If the following two conditions hold, then add $U_i$ to $\cP_i$:
\begin{align} \label{eq:cond-deg-concen}
\left\|\ow(U_i) - \ow(Q)\right\|_\infty \leq 2\gamma',
\end{align}
\begin{align} \label{eq:cond-opt-concen}
|\oind(U_i)^T \ow(Q) - \oind(Q)^T \ow(Q)| \leq 4\gamma'.
\end{align}
\end{itemize}
\end{itemize}
\item Create a partition matroid $\cM$ on the ground set $\cP_1 \cup \cdots \cup \cP_s$ such that $\cS \subseteq \cP_1 \cup \cdots \cup \cP_s$ is an independent set iff $|\cS \cap \cP_i| \leq 1$ for all $i \in [s]$.
\item Let $f$ denote the set function on the ground set $\cP_1 \cup \cdots \cup \cP_s$ defined as $h(\cS) := h\left(I \cup \bigcup_{S \in \cS} S\right)$
\item Run the algorithm from \Cref{thm:submodular-matroid} to on $(f, \cM)$ to get a set $\cZ_Q \subseteq \cP_1 \cup \cdots \cup \cP_s$.
\item Let $\tZ_Q = \bigcup_{S \in \cZ} S$.
\item If $|\tZ_Q| \geq k'$, let $\tT_Q$ be a random subset of $\tZ_Q$ of size $k'$. Otherwise, let $Z_Q$ be an arbitrary superset of $\tZ_Q$ of size $k'$.
\end{itemize}
\item Output the best set $I \cup Z_Q$ found among all $Q$'s.
\end{itemize}

It is obvious to see that the algorithm runs in $n^{O(t)} = n^{O(\log n / \gamma^2)}$ time. The rest of the proof is devoted to proving~\eqref{eq:apx-guarantee-submodular-dks}.

Let $T'^{\OPT} := T^{\OPT} \setminus I$, and let $U_i^{\OPT} := V'_i \cap T^{\OPT}$. We will start by proving the following claim, which (as we will argue below) ensures that w.h.p. $U_i^{\OPT}$ is included in $\cP_i$.

\begin{claim} \label{claim:random-partition-conditions}
With probability $1 - O(1/n)$ (over the random partition $V'_1, \dots, V'_s$), the following holds for all $i \in [s]$:
\begin{align} \label{eq:intersection-size}
|U_i^{\OPT}| \in [(1 - \gamma')t, (1 + \gamma')t]
\end{align}
\begin{align} \label{eq:weight-apx}
\left\|\ow(U_i) - \ow(T'^{\OPT})\right\|_\infty \leq \gamma'
\end{align}
\begin{align} \label{eq:opt-large}
|\oind(U_i)^T \ow(T'^{\OPT}) - \oind(T'^{\OPT})^T \ow(T'^{\OPT})| \leq \gamma'
\end{align}
\end{claim}

\begin{proof}[Proof of \Cref{claim:random-partition-conditions}]
We will argue that all conditions holds for a fixed $i \in [s]$ with probability $1 - O(1/n^2)$. Union bound over all $i \in [s]$ then yields the claim. 

Let us fix $i \in [s]$. Since each vertex is included in $U_i$ with probability $1/s$, we may apply Chernoff bound (\Cref{lem:chernoff}) to conclude that
\begin{align} \label{eq:intersection-size-fixed-i}
\Pr[|U_i^{\OPT}| \notin [(1 - \gamma')t, (1 + \gamma')t]] \leq 2 \exp\left(-\frac{\gamma'^2 t}{3}\right) \leq 2/n^3.
\end{align}

Next, consider a fixed $v \in V$. We will now bound the probability that $|\ow(U_i)_v - \ow(T'^{\OPT})_v| < \gamma'$. To do so, let us condition on the size of $U_i^{\OPT}$ equal to $g \in \N$. After such a conditioning, we may view the set $U_i^{\OPT}$ as being generated by drawing $u_1, \dots, u_g$ randomly without replacement from $T'^{\OPT}$. Since $\ow(U_i^{\OPT})_v = \frac{1}{g}\left(\sum_{i \in [g]} w(\{u_i, v\})\right)$ and $\E[\ow(U_i^{\OPT})_v] = \ow(T'^{\OPT})_v$, we may apply \Cref{lem:hoeffding} to conclude that
\begin{align} \label{eq:deg-concen-fixed-i}
\Pr[|\ow(U_i^{\OPT})_v - \ow(T'^{\OPT})_v| > \gamma' \mid |U_i^{\OPT}| = g] \leq 2\exp(-\gamma'^2 g). 
\end{align}
Therefore, we have
\begin{align*}
&\Pr[|\ow(U_i^{\OPT})_v - \ow(T'^{\OPT})_v| > \gamma'] \\
&\leq \Pr[|U_i^{\OPT}| < (1 - \gamma')t] + \Pr[|\ow(U_i^{\OPT})_v - \ow(T'^{\OPT})_v| > \gamma' \mid |U_i^{\OPT}| \geq (1 - \gamma')t] \\
&\overset{\text{\eqref{eq:intersection-size-fixed-i}, \eqref{eq:deg-concen-fixed-i}}}{\leq} 2/n^3 + 2\exp(-\gamma'^2 (1 - \gamma')t) \\
&\leq 4/n^3.
\end{align*}
Taking the union bound over all $v \in V$, we have
\begin{align*}
\Pr[\|\ow(U_i^{\OPT})_v - \ow(T'^{\OPT})_v\|_{\infty} > \gamma'] \leq 4/n^2.
\end{align*}

Analogous arguments also imply that
\begin{align*}
\Pr[|\oind(U_i)^T \ow(T'^{\OPT}) - \oind(T'^{\OPT})^T \ow(T'^{\OPT})| > \gamma'] \leq O(1/n^2).
\end{align*}

Applying the union bound, we conclude that all three conditions hold for a fixed $i$ with probability at least $1 - O(1/n^2)$. Finally, applying the union bound over all $i \in [s]$, we have that all three conditions hold for all $i \in [s]$ with probability at least $1 - O(1/n)$, which concludes our proof.
\end{proof}

Let $E$ denote the event that all conditions in \Cref{claim:random-partition-conditions} hold for all $i$. Conditioned on $E$, and letting $Q = U_1$. For all $i \in [s]$, we have
\begin{align*}
\left\|\ow(U_i) - \ow(Q)\right\|_\infty \leq \left\|\ow(U_i) - \ow(T'^{\OPT})\right\|_\infty + \left\|\ow(T'^{\OPT}) - \ow(Q)\right\|_\infty \overset{\text{\eqref{eq:weight-apx}}}{\leq} 2\gamma'.
\end{align*}
and
\begin{align*}
&|\oind(U_i)^T \ow(Q) - \oind(Q)^T \ow(Q)| \\
&\leq |\oind(U_i)^T \ow(T'^{\OPT}) - \oind(T'^{\OPT})^T \ow(T'^{\OPT})| + |\oind(T'^{\OPT})^T \ow(T'^{\OPT}) - \oind(Q)^T \ow(T'^{\OPT})| \\
&\qquad + |\oind(U_i)^T(\ow(Q) - \ow(T'^{\OPT}))| + |\oind(Q)^T(\ow(T'^{\OPT}) - \ow(Q))| \\
&\overset{\text{\eqref{eq:opt-large}}}{\leq} 2\gamma' + \|\oind(U_i)\|_1 \|\ow(Q) - \ow(T'^{\OPT})\|_\infty + \|\oind(Q)\|_1 \|\ow(Q) - \ow(T'^{\OPT})\|_\infty \\
&\overset{\text{\eqref{eq:weight-apx}}}{\leq} 4\gamma'.
\end{align*}
Therefore, $U_i^{\OPT}$ is included in $\cP_i$ for all $i \in [s]$. Thus, the guarantee of \Cref{thm:submodular-matroid} means that $\E[f(\cZ_{Q})] \geq (1 - 1/e)f(\{U_1^{\OPT}, \dots, U_s^{\OPT}\})$. This is equivalent to 
\begin{align} \label{eq:submod-not-final}
\E[h(I \cup \tZ_Q)] \geq h(T^{\OPT}).
\end{align}
Next, we will lower bound $\den(\tZ_Q)$. Note that we may assume\footnote{If $\cZ_Q \cap \cP_i = \emptyset$, we may add any element of $\cP_i$ to $\cZ_Q$.} that $\cZ_Q = \{U_1, \dots, U_s\}$ where $U_i \in \cP_i$. 

For sets $A, B \subseteq V$, we write $w(A)$ to denote $\sum_{\{u, v\} \subseteq A} w(\{u, v\})$ and $w(A, B)$ to denote $\sum_{u \in A, v \in B} w(\{u, v\})$.
We may write $w(I \cup \tZ_Q)$ as
\begin{align}
w(I \cup \tZ_Q)
&= w(I \cup U_1 \cup \dots \cup U_s) \nonumber \\
&= w(I) + \sum_{i \in [s]} w(I, U_i) + \frac{1}{2} \sum_{i, j \in [s]} w(U_i, U_j). \label{eq:weight-expand}
\end{align}
We will now lower bound each term in the sum. First, we have
\begin{align*}
w(I, U_i) &= \ind(I)^T W \ind(U_i) \\
&= \ind(I)^T \bw(U_i) \\
&= |I| \cdot |U_i| \cdot \oind(I)^T \ow(U_i) \\
&\overset{\eqref{eq:cond-deg-concen}}{\geq} |I| \cdot |U_i| \cdot \left(\oind(I)^T \ow(Q) - 2\gamma'\right) \\
&\overset{\eqref{eq:weight-apx}}{\geq} |I| \cdot |U_i| \cdot \left(\oind(I)^T \ow(T'^{\OPT}) - 3\gamma'\right).
\end{align*}
Recall from our construction that $|U_i| \geq (1 - \gamma')t = (1 - \gamma') |T'^{\OPT}|/s$. Plugging this into the above, we have
\begin{align} \label{eq:ui-i-lb}
w(I, U_i) \geq \frac{1 - \gamma'}{s} \cdot w(I, T'^{\OPT}) - 3\gamma' \cdot |I| \cdot |U_i|.
\end{align}
 
Secondly, we have
\begin{align*}
w(U_i, U_j) &= \ind(U_i)^T W \ind(U_j) \\
&= \ind(U_i)^T \bw(U_j) \\
&= |U_i| \cdot |U_j| \cdot \oind(U_i)^T \ow(U_j) \\
&\overset{\eqref{eq:cond-deg-concen}}{\geq} |U_i| \cdot |U_j| \cdot \left(\oind(U_i)^T \ow(Q) - 2\gamma'\right) \\
&\overset{\eqref{eq:cond-opt-concen}}{\geq} |U_i| \cdot |U_j| \cdot \left(\oind(Q)^T \ow(Q) - 6\gamma'\right) \\
&\overset{\eqref{eq:weight-apx}}{\geq} |U_i| \cdot |U_j| \cdot \left(\oind(Q)^T \ow(T'^{\OPT}) - 7\gamma'\right) \\
&\overset{\eqref{eq:opt-large}}{\geq} |U_i| \cdot |U_j| \cdot \left(\oind(T'^{\OPT})^T \ow(T'^{\OPT}) - 8\gamma'\right) \\
&= |U_i| \cdot |U_j| \cdot \left(\frac{2}{|T'^{\OPT}|^2} \cdot w(T'^{\OPT}) - 8\gamma'\right)
\end{align*}
Similar to above, we may now use the fact that $|U_i|, |U_j| \geq (1 - \gamma') |T'^{\OPT}|/s$ to derive
\begin{align} \label{eq:ui-uj-lb}
w(U_i, U_j) \geq \frac{2(1 - 2\gamma')}{s^2} w(T'^{\OPT}) - 8\gamma' \cdot |U_i| \cdot |U_j|
\end{align}

Plugging~\eqref{eq:ui-i-lb} and~\eqref{eq:ui-uj-lb} back into~\eqref{eq:weight-expand}, we arrive at
\begin{align*}
w(I \cup \tZ_Q) 
&\geq w(I) + (1 - \gamma') w(I, T'^{\OPT}) + (1 - 2\gamma') w(T'^{\OPT}) - 8 \gamma' |I \cup \tZ_Q|^2 \\
&\geq (1 - 2\gamma') w(S^{OPT}) - 8 \gamma' |I \cup \tZ_Q|^2.
\end{align*}
Therefore,
\begin{align*}
\den(I \cup \tZ_Q) = \frac{1}{\binom{|I \cup \tZ_Q|}{2}} \cdot w(I \cup \tZ_Q) \geq \frac{1 - 2\gamma'}{\binom{|I \cup \tZ_Q|}{2}} \cdot w(T^{\OPT}) - 16\gamma'.
\end{align*}
Since $|U_i| \leq (1 + \gamma')t$, we have $|\tZ_Q| \leq (1 + \gamma')k'$. This implies that $|I \cup \tZ_Q| \leq (1 + \gamma')k$ Therefore,
\begin{align} 
\den(I \cup \tZ_Q) 
&\geq \frac{k(k - 1)}{(1+\gamma')k((1+\gamma')k + 1)} \cdot \frac{1 - 2\gamma'}{\binom{k}{2}} \cdot w(T^{\OPT}) - 16\gamma' \nonumber \\
&\geq (1 - 5\gamma') \den(T^{\OPT}) - 16\gamma'. \label{eq:den-almost-final}
\end{align}

From \eqref{eq:submod-not-final} and \eqref{eq:den-almost-final}, we can conclude that 
\begin{align}
\E[\den(\tZ_Q \cup I) + h(\tZ_Q \cup I) | E] \geq (1 - 5\gamma') \den(T^{\OPT}) - 16\gamma' + (1 - 1/e)h(T^{\OPT}).
\end{align}
Recall from \Cref{claim:random-partition-conditions} that $E$ happens with probability at least $1 - O(1/n)$. Thus, we have
\begin{align}
&\E[\den(\tZ_Q \cup I) + h(\tZ_Q \cup I)] \nonumber \\
&\geq  (1 - 5\gamma' - O(1/n)) \den(T^{\OPT}) - 16\gamma' + (1 - 1/e - O(1/n))h(T^{\OPT}). \label{eq:objective-bound-almost-final}
\end{align}

Finally, $|\tZ_Q| \leq (1 + \gamma')k'$ also implies that
\begin{align*}
&\E[\den(Z_Q \cup I) + h(Z_Q \cup I)] \\
&\geq \frac{k'(k' - 1)}{(1 + \gamma')k'((1 + \gamma')k' - 1)} \E[\den(\tZ_Q \cup I)] + \frac{k'}{(1 + \gamma')k'} \E[h(\tZ_Q \cup I)] \\
&\geq (1 - 3\gamma') \E\left[\left(\den(\tZ_Q \cup I) + h(\tZ_Q \cup I) \right)\right] \\
&\overset{\eqref{eq:objective-bound-almost-final}}{\geq} (1 - 8\gamma' - O(1/n)) \den(T^{\OPT}) - 16\gamma' + (1 - 1/e - 3\gamma' - O(1/n))h(T^{\OPT}) \\
&\geq \den(T^{\OPT}) + (1 - 1/e - 3\gamma' - O(1/n))h(T^{\OPT}) - 24\gamma' - O(1/n),
\end{align*}
which is at least $\den(T^{\OPT}) + (1 - 1/e - \gamma)h(T^{\OPT}) - \gamma$ for sufficiently large $n \geq \Omega(1/\gamma)$. Note that when $n$ is $O(1/\gamma)$, we may simply run the bruteforce $2^{O(n)} = 2^{O(1/\gamma)}$ to solve the problem.
\end{proof}

\subsection{From Submodular DkS to Max-Sum Diversification}

Having provided an approximation algorithm for Submodular \dks, we now turn our attention back to how to use it to approximate Max-Sum Diversification.

\subsubsection{A Structural Lemma}

We start by proving a structural lemma for Max-Sum Diversification that is analogous to \Cref{lem:max-dispersion-structural} for Max-Sum Dispersion.

\begin{lemma} \label{lem:max-diversification-structural}
Let $S^{\OPT}$ be any optimal solution of Max-Sum Diversification and let $u^{\min}$ be the vertex in $S^{\OPT}$ that minimizes $\disp(u^{\min}, S^{\OPT})$. Furthermore, let $v$ be any vertex \emph{not} in $S^{\OPT}$ and let $\Delta = d(u^{\min}, v)$. Then, we have
\begin{align*}
\dive(S^{\OPT}) \geq \frac{p(p - 1)\Delta}{16}.
\end{align*} 
\end{lemma}

\begin{proof}[Proof of \Cref{lem:max-diversification-structural}]
Let $\SOPTc := S^{\OPT} \cap \cB{u^{\min}}{0.5\Delta}$. Consider two cases, based on the size of $\SOPTc$:
\begin{itemize}
\item Case I: $|\SOPTc| \leq p / 2$. This is similar to the first case in the proof of \Cref{lem:max-dispersion-structural}: we have 
\begin{align*}
\disp(u^{\min}, S^{\OPT}) \geq \disp(u^{\min}, S^{\OPT} \setminus \SOPTc) \geq (p/2)(\Delta/2) = \Delta p / 4.
\end{align*}
Furthermore, by our definition of $u^{\min}$, we have
\begin{align*}
\disp(S^{\OPT}) = \frac{1}{2} \sum_{u \in S} \disp(u, S^{\OPT}) \geq \frac{p}{2} \disp(u^{\min}, S^{\OPT}).
\end{align*}
Combining the two inequalities, we have $\dive(S^{\OPT}) \geq \disp(S^{\OPT}) \geq p^2\Delta / 8$.
\item Case II: $|\SOPTc| > p / 2$. In this case, since $S^{\OPT}$ is an optimal solution, replacing any $z \in S^{\OPT}_{\text{close}}$ with $v$ must not increase the solution value, i.e. 
\begin{align*}
\left[f(S^{\OPT}) - f(S^{\OPT} \setminus \{z\})\right] + \disp(z, S^{\OPT}) &\geq \disp(v, S^{\OPT} \setminus \{z\}) \\
&\geq \disp(v, \SOPTc \setminus \{z\}) \\
&\geq ((p - 1)/2)(0.5\Delta),
\end{align*}
where the second inequality uses the fact that for any $z' \in \SOPTc$ we have $d(v, z') \geq d(u, v) - d(u, z') \geq \Delta - 0.5\Delta$. From this, we have
\begin{align*}
&\dive(S^{\OPT}) = f(S^{\OPT}) + \disp(S^{\OPT}) \\
&= f(S^{\OPT}) + \frac{1}{2} \sum_{u \in S} \disp(u, S^{\OPT}) \\
&\geq \sum_{z \in S^{\OPT}} \left[f(S^{\OPT}) - f(S^{\OPT} \setminus \{z\})\right] + \frac{1}{2} \sum_{u \in S} \disp(u, S^{\OPT}) \\
&\geq \frac{1}{2} \sum_{z \in \SOPTc} \left(\left[f(S^{\OPT}) - f(S^{\OPT} \setminus \{z\})\right] + \disp(z, S^{\OPT})\right) \\
&\geq |\SOPTc| \cdot \frac{(p - 1)\Delta}{8} \\
&> \frac{p(p - 1)}{16\Delta},
\end{align*}
where the last inequality follows from our assumption of this case. \qedhere
\end{itemize}
\end{proof}

\subsubsection{Putting Things Together: Proof of \Cref{thm:diversification-main-detailed}}

Finally, we use the structural lemma to reduce Max-Sum Diversification to Submodular \dks. Again, this reduction is analogous to that of Max-Sum Dispersion to \dks\ presented in the previous section.

\begin{proof}[Proof of \Cref{thm:diversification-main-detailed}]
Our algorithm works as follows:
\begin{enumerate}
\item For every distinct $u, v \in U$ do:
\begin{enumerate}
\item Let $\Delta := d(u, v)$ and $\Delta^* = 20 \Delta / \eps$.
\item If $|\bcB{u}{\Delta}| \geq p$, then skip the following steps and continue to the next pair $u, v$.
\item Otherwise, create a submodular \dks\ instance where $V := \cB{u}{\Delta^*}, I := V \setminus \cB{u}{\Delta}$, $k = p - |\bcB{u}{\Delta^*}|$, define $h$ by $h(C) := \frac{1}{k(k - 1) \Delta^*} \cdot f(\bcB{u}{\Delta} \cup C)$, and define $w$ by $w(\{y, z\}) := (0.5 / \Delta^*) \cdot d(y, z)$ for all $y, z \in V$.
\item Use the algorithm from \Cref{thm:submodular-dks} to solve the above instance with $\gamma := 0.00005\eps^2$. Let $T$ be the solution found.
\item Finally, let $S^{u, v} := T \cup \bcB{u}{\Delta^*}$.
\end{enumerate}
\item Output the best solution among $S^{u, v}$ considered.
\end{enumerate}

It is obvious that the running time is dominated by the running time of the algorithm from \Cref{thm:submodular-dks} which takes $n^{O(\log n / \gamma^2)} = n^{O(\log n / \eps^4)}$ as desired.

Next, we prove the algorithm's approximation guarantee. To do this, let us consider $S^{\OPT}, u^{\min}$ as defined in \Cref{lem:max-dispersion-structural}, and let $u = u^{\min}, v := \argmax_{z \notin S^{\OPT}} d(u, z), \Delta = d(u, v)$. Recall that by the definition of $v$, we have $S^{\OPT} \supseteq \bcB{u}{\Delta} = \bcB{u}{\Delta^*} \cup I$. 
Let $T$ be the solution found by the submodular \dks\ algorithm for this $u, v$ and let $T' := T \setminus I$. We have
\begin{align}
&\disp(S^{u, v}) \nonumber \\
&= \disp(\bcB{u}{\Delta^*}) + \disp(\bcB{u}{\Delta^*}, T) + \disp(T) \nonumber \\
&= \disp(\bcB{u}{\Delta^*}) + \disp(\bcB{u}{\Delta^*}, I) + \disp(\bcB{u}{\Delta^*}, T') + \disp(T). \label{eq:current-s-decompose-2}
\end{align}
Similarly, letting $S := S^{\OPT} \cap \cB{u}{\Delta^*}$ and $S' := S^{\OPT} \setminus I$, we have
\begin{align}
&\disp(S^{\OPT}) \nonumber \\
&= \disp(\bcB{u}{\Delta^*}) + \disp(\bcB{u}{\Delta^*}, I) + \disp(\bcB{u}{\Delta^*}, S') + \disp(S). \label{eq:optimal-s-decompose-2}
\end{align}

Now, observe from the definition of the submodular \dks\ instance (for this $u, v$) that for any $J$ such that $I \subseteq J \subseteq V$, we have
\begin{align*} 
\den(J) = \frac{1}{k(k - 1)\Delta^*} \cdot \disp(J)
\end{align*}
and
\begin{align*}
h(J) = \frac{1}{k(k - 1)\Delta^*} \cdot f(\bcB{u}{\Delta} \cup J).
\end{align*}
The approximation guarantee from \Cref{thm:submodular-dks} ensures that $\E[h(S^{u, v}) + \den(T)] \geq (1 - 1/e - \gamma)h(S^{\OPT}) + \den(S) - \gamma$. Using the above two equalities, we can rewrite this guarantee as
\begin{align} \label{eq:dks-submodular-guarantee}
\disp(S) + (1 - 1/e - \gamma)f(S^{\OPT}) - \E[\disp(T) + f(S^{u, v})] \leq \gamma \cdot \Delta^* \cdot k(k - 1).
\end{align}

Taking the difference between~\eqref{eq:optimal-s-decompose-2} and~\eqref{eq:current-s-decompose-2} and applying~\eqref{eq:dks-submodular-guarantee}, we have
\begin{align*}
&\disp(S^{\OPT}) + (1 - 1/e - \gamma)f(S^{\OPT}) - \E[\disp(S^{u, v}) + f(S^{u, v})] \\
&\leq \disp(\bcB{z}{\Delta^*}, S') - \disp(\bcB{z}{\Delta^*}, T') + \gamma \cdot 2\Delta^* \cdot k(k - 1). \\
(\text{Our choice of } \gamma) &\leq \disp(\bcB{z}{\Delta^*}, S') - \disp(\bcB{z}{\Delta^*}, T') + 0.001\eps \Delta \cdot p(p - 1) \\
(\text{\Cref{lem:max-diversification-structural}}) &\leq \disp(\bcB{z}{\Delta^*}, S') - \disp(\bcB{z}{\Delta^*}, T') + 0.1\eps \dive(S^{\OPT}).
\end{align*}

Now, since $|S'| = |T'| \leq p$ and $S', T' \subseteq \cB{z}{\Delta}$, we have
\begin{align*}
\disp(\bcB{z}{\Delta^*}, S') - \disp(\bcB{z}{\Delta^*}, T') 
&\leq |\bcB{z}{\Delta^*}| \cdot |S'| \cdot ((\Delta^* + \Delta) - (\Delta^* - \Delta)) \\
&\leq 2 |\bcB{z}{\Delta^*}| \cdot |S'| \cdot \Delta \\
(\text{Our choice of } \Delta^*) &\leq 0.1\eps \cdot |\bcB{z}{\Delta^*}| \cdot |S'| \cdot (\Delta^* - \Delta) \\
&\leq 0.1\eps \cdot \disp(\bcB{z}{\Delta^*}, S') \\
&\leq 0.1\eps \cdot \disp(S^{\OPT}).
\end{align*}

Combining the above two inequalities, we get
\begin{align*}
\E[\dive(S^{u, v})] 
&\geq \disp(S^{\OPT}) + (1 - 1/e - \gamma)f(S^{\OPT}) - 0.2\eps \dive(S^{\OPT}) \\
&\geq (1 - \eps) \disp(S^{\OPT}) + (1 - 1/e - \eps) f(S^{\OPT}),
\end{align*}
completing our proof.
\end{proof}
\fi 

\end{document}